\documentclass{IEEEtaes}
\usepackage{cite}
\usepackage{amsmath,amssymb,amsfonts}
\usepackage{graphicx}
\usepackage{textcomp}
\usepackage{xcolor}
\usepackage{color,array,amsthm, amsmath, amsfonts, amssymb, bbm, url,array, comment}
\usepackage{graphicx}
\usepackage{subcaption}
\usepackage{url, hyperref}
\usepackage[acronyms]{glossaries}
\usepackage{mwe}
\usepackage{tikz}
\usepackage{mathtools} 
\usepackage[ruled,vlined]{algorithm2e}

\usetikzlibrary{positioning, arrows.meta}

\makeglossaries
\newtheorem{theorem}{Theorem}
\newacronym{lstm}{LSTM}{Long Short Term Memory}
\newacronym{gnn}{GNN}{Graph Neural Network}
\newacronym{map}{MAP}{Maximum Aposteriori}
\newacronym{cyk}{CYK}{Cocke-Younger-Kasami}
\newacronym{scsg}{SCSG}{Stochastic Context Sensitive Grammar}
\newacronym{bert}{BERT}{Bidirectional Encoder Representations from Transformers}
\newacronym{gtnn}{GTNN}{Graph Transformer Neural Network}
\newacronym{nlp}{NLP}{Natural Language Processing}

\newcommand{\alphabet}{\mathcal{A}}

\newcommand{\state}{\mathbf{x}}
\newcommand{\stateinf}{\mathbf{\hat{x}}}
\newcommand{\statenoise}{\mathbf{w}}
\newcommand{\observation}{\mathbf{y}}
\newcommand{\observationNoise}{\mathbf{o}}
\newcommand{\control}{\mathbf{a}}

\newcommand{\utility}{u}
\RequirePackage[normalem]{ulem} 
\RequirePackage{color}\definecolor{RED}{rgb}{1,0,0}\definecolor{BLUE}{rgb}{0,0,1} 

\providecommand{\DIFadd}[1]{#1}
\providecommand{\DIFdel}[1]{}
\providecommand{\DIFaddbegin}{} 
\providecommand{\DIFaddend}{} 
\providecommand{\DIFdelbegin}{} 
\providecommand{\DIFdelend}{} 
\providecommand{\DIFaddFL}[1]{\DIFadd{#1}} 
\providecommand{\DIFdelFL}[1]{\DIFdel{#1}} 
\providecommand{\DIFaddbeginFL}{} 
\providecommand{\DIFaddendFL}{} 
\providecommand{\DIFdelbeginFL}{} 
\providecommand{\DIFdelendFL}{} 
\newcommand{\DIFscaledelfig}{0.5}
\RequirePackage{settobox} 
\RequirePackage{letltxmacro} 
\newsavebox{\DIFdelgraphicsbox} 
\newlength{\DIFdelgraphicswidth} 
\newlength{\DIFdelgraphicsheight} 
\LetLtxMacro{\DIFOincludegraphics}{\includegraphics} 
\newcommand{\DIFaddincludegraphics}[2][]{{\color{blue}\fbox{\DIFOincludegraphics[#1]{#2}}}} 
\newcommand{\DIFdelincludegraphics}[2][]{
\sbox{\DIFdelgraphicsbox}{\DIFOincludegraphics[#1]{#2}}
\settoboxwidth{\DIFdelgraphicswidth}{\DIFdelgraphicsbox} 
\settoboxtotalheight{\DIFdelgraphicsheight}{\DIFdelgraphicsbox} 
\scalebox{\DIFscaledelfig}{
\parbox[b]{\DIFdelgraphicswidth}{\usebox{\DIFdelgraphicsbox}\\[-\baselineskip] \rule{\DIFdelgraphicswidth}{0em}}\llap{\resizebox{\DIFdelgraphicswidth}{\DIFdelgraphicsheight}{
\setlength{\unitlength}{\DIFdelgraphicswidth}
\begin{picture}(1,1)
\thicklines\linethickness{2pt} 
{\color[rgb]{1,0,0}\put(0,0){\framebox(1,1){}}}
{\color[rgb]{1,0,0}\put(0,0){\line( 1,1){1}}}
{\color[rgb]{1,0,0}\put(0,1){\line(1,-1){1}}}
\end{picture}
}\hspace*{3pt}}} 
} 
\LetLtxMacro{\DIFOaddbegin}{\DIFaddbegin} 
\LetLtxMacro{\DIFOaddend}{\DIFaddend} 
\LetLtxMacro{\DIFOdelbegin}{\DIFdelbegin} 
\LetLtxMacro{\DIFOdelend}{\DIFdelend} 
\DeclareRobustCommand{\DIFaddbegin}{\DIFOaddbegin \let\includegraphics\DIFaddincludegraphics} 
\DeclareRobustCommand{\DIFaddend}{\DIFOaddend \let\includegraphics\DIFOincludegraphics} 
\DeclareRobustCommand{\DIFdelbegin}{\DIFOdelbegin \let\includegraphics\DIFdelincludegraphics} 
\DeclareRobustCommand{\DIFdelend}{\DIFOaddend \let\includegraphics\DIFOincludegraphics} 
\LetLtxMacro{\DIFOaddbeginFL}{\DIFaddbeginFL} 
\LetLtxMacro{\DIFOaddendFL}{\DIFaddendFL} 
\LetLtxMacro{\DIFOdelbeginFL}{\DIFdelbeginFL} 
\LetLtxMacro{\DIFOdelendFL}{\DIFdelendFL} 
\DeclareRobustCommand{\DIFaddbeginFL}{\DIFOaddbeginFL \let\includegraphics\DIFaddincludegraphics} 
\DeclareRobustCommand{\DIFaddendFL}{\DIFOaddendFL \let\includegraphics\DIFOincludegraphics} 
\DeclareRobustCommand{\DIFdelbeginFL}{\DIFOdelbeginFL \let\includegraphics\DIFdelincludegraphics} 
\DeclareRobustCommand{\DIFdelendFL}{\DIFOaddendFL \let\includegraphics\DIFOincludegraphics} 
\RequirePackage{listings} 
\RequirePackage{color} 
\lstdefinelanguage{DIFcode}{ 
  moredelim=[il][\color{red}\sout]{\%DIF\ <\ }, 
  moredelim=[il][\color{blue}\uwave]{\%DIF\ >\ } 
} 
\lstdefinestyle{DIFverbatimstyle}{ 
	language=DIFcode, 
	basicstyle=\ttfamily, 
	columns=fullflexible, 
	keepspaces=true 
} 
\lstnewenvironment{DIFverbatim}{\lstset{style=DIFverbatimstyle}}{} 
\lstnewenvironment{DIFverbatim*}{\lstset{style=DIFverbatimstyle,showspaces=true}}{} 

\begin{document}
\title{Inferring Group Intent as a Cooperative Game. \\
An NLP-based Framework for Trajectory Analysis 
\DIFdelbegin \DIFdel{using Graph Transformer Neural Network}\DIFdelend 
}
\author{Yiming Zhang}
\affil{Cornell University, Ithaca, NY, USA} 

\author{Vikram Krishnamurthy}
\affil{Cornell University, Ithaca, NY, USA} 

\author{Shashwat Jain}
\affil{Cornell University, Ithaca, NY, USA} 

\DIFdelbegin 
\DIFdelend \DIFaddbegin \authoraddress{Yiming Zhang: yz2926@cornell.edu, Vikram Krishnamurthy: vikramk@cornell.edu, Shashwat Jain: sj474@cornell.edu.}

\editor{}
\supplementary{This research was supported by NSF grants CCF-2312198 and CCF-2112457 and Army Research grant W911NF-24-1-0083.}

\markboth{Zhang et.al}{Group Intent Inference using GTNN}
\DIFaddend 

\maketitle

\begin{abstract}
This paper studies group target trajectory intent as the outcome of a cooperative game where the \DIFdelbegin \DIFdel{complex-spatio }\DIFdelend \DIFaddbegin \DIFadd{complex spatio-temporal }\DIFaddend trajectories are modeled using an NLP-based generative model. 
In our framework, the group intent is specified by the characteristic function of a cooperative game, and allocations for \DIFdelbegin \DIFdel{player }\DIFdelend \DIFaddbegin \DIFadd{players }\DIFaddend in the cooperative game are specified by either the core, the Shapley value, or the nucleolus. 
The resulting allocations induce probability distributions that govern the coordinated spatio-temporal trajectories of the targets that reflect the group’s underlying intent.
We address two key questions: (1) How can the intent of a group trajectory be optimally formalized as the characteristic function of a cooperative game? (2) How can such intent be inferred from noisy observations of the targets?
To answer the first question, we introduce a Fisher-information-based characteristic function of the cooperative game\DIFaddbegin \DIFadd{, }\DIFaddend which yields probability distributions that generate coordinated spatio-temporal patterns. As a generative model for these patterns, we develop an NLP-based generative model built on formal grammar, enabling the creation of realistic multi-target trajectory data. To answer the second question,
we train a Graph Transformer Neural Network (GTNN) to infer group trajectory intent—expressed as the characteristic function of the cooperative game—from observational data with high accuracy. The self-attention function of the GTNN depends on the track estimates. Thus\DIFaddbegin \DIFadd{, }\DIFaddend the formulation and algorithms provide a multi-layer approach that spans  target tracking  (Bayesian signal processing) and the GTNN (for group intent inference).

\end{abstract}



\DIFdelend \section{Introduction}

This paper studies group trajectory intent as the outcome of a cooperative game, in which complex spatio-temporal trajectories are modeled using an NLP-based generative approach. 
Here, group intent denotes the underlying group objective or coordinated pattern that governs the trajectories of multiple targets—beyond their individual motions.
Inferring this intent enables trackers to more accurately anticipate future group behavior. Modeling group intent with cooperative game theory provides a principled way to capture target coordination. We infer group intent by using a transformer-based classifier inspired by \gls*{bert}, where trajectories are cast into a language-like representation as input. \gls*{bert} is a transformer-based language model that learns bidirectional representations of text, enabling accurate performance on diverse natural language processing tasks.

While multi-target tracking is a mature area, understanding the intent of a group of targets has received relatively little attention. 
To capture group target intent, classical radar-based tracking methods typically rely on Markov state-space models to represent target kinematics. These models are effective over short time horizons and have led to the development of numerous tracking algorithms in the literature \cite{bar2004estimation, xie1990multiple, oh2009markov}.
This paper is motivated by metalevel tracking on longer timescales. 
In metalevel tracking, in which one is interested in devising automated procedures that assist a human analyst to interpret the tracks obtained from a conventional tracking algorithm. 
On such longer timescales, real-world targets are driven by a premeditated intent. The intent of a group of targets is reflected in the characteristic function of the cooperative game that they participate in.

This paper is motivated by metalevel tracking, in which one seeks to devise automated procedures that assist a human analyst in interpreting the tracks obtained from a conventional tracking algorithm. At such longer timescales, real-world targets are often driven by premeditated intent. In our framework, the intent of a group of targets is reflected in the characteristic function of the cooperative game in which they participate.


\begin{figure}[h!]
    \centering
    \includegraphics[width=0.9\linewidth]{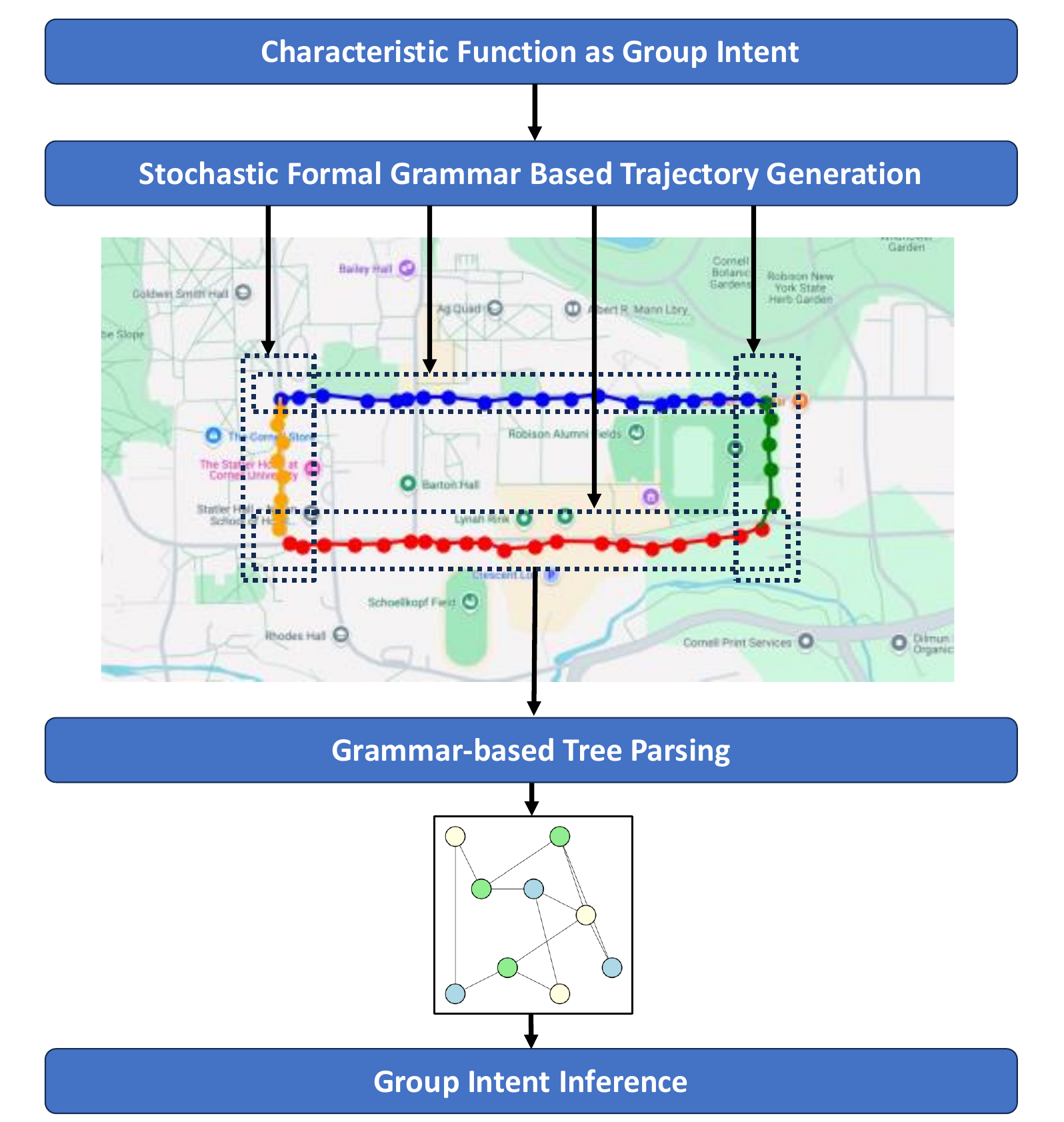}
    \caption{Overview of the proposed framework. A team of targets forms a rectangular surveillance pattern, modeled via a cooperative game and stochastic grammar, with a GTNN recovering the group’s intent from the parse tree.}
    \label{fig:teaser}
\end{figure}

{\bf Example}. Fig.~\ref{fig:teaser}, shows  a group of targets whose intent is to surveil an area in a rectangular pattern of arbitrary size. Collectively, their trajectories form the shape of a rectangle. 
We use the characteristic function of a cooperative game as a generative model for how the targets decide which part of the rectangle each target should surveil. 
To model such trajectory shapes, we employ a natural language-inspired approach based on stochastic formal grammars, which serve as generative models capable of capturing complex spatio-syntactic patterns \cite{ferrucci1994symbol,zhang2001context,kong2006spatial,balari2012knots}. Finally, to recover the group intent (characteristic function of the cooperative game), we develop a GTNN architecture that exploits the structure of the parse tree.

\subsection{Cooperative Game as a Generative Model for Group Intent}
\label{subsec:cooperative_game_motivation}

In multi-target tracking scenarios such as team-based navigation, agents act not only on individual goals but also in coordination to achieve a group-level intent. This intent is expressed through the group trajectory, describing how the group evolves in space and time. We formalize this idea using a cooperative game-theoretic framework, treating each target as a player in a game defined over possible trajectories. The framework ensures subgroup fairness by capturing  individual preferences and collective cohesion.

We model cooperation amongst targets via  the characteristic function  which quantifies the achievable utility (e.g., coverage or surveillance quality) of each coalition. The characteristic function  serves as a generative model for group intent, mapping coalitional structure to expected outcomes. To ground this model, we introduce a Fisher–information–based characteristic function that defines the probability distribution over the group’s objective. The allocation solutions of the cooperative game, namely, the core, nucleolus, and Shapley value~\cite{schmeidler1969nucleolus, shapley1953value}, capture stability, fairness, and marginal contributions respectively.

\subsection{Transformer based Intent Inference Architecture}
Building on the above generative model for group intent,  we represent group intent with a stochastic context-sensitive grammar (SCSG), where production rule probabilities  are determined  by cooperative-game allocations. SCSGs balance expressiveness and tractability, though inference is NP-hard. To address this, we represent trajectories as parse trees and employ \gls*{gtnn}~\cite{GCN1, GCN2} for efficient inference, enabling reliable modeling of group intent.

We design a \gls*{bert}-inspired architecture for intent inference from trajectory graphs.
\gls*{bert}, is a deep learning model that applies a self-attention mechanism to learn the contextual relationships between data points in a sequence from both directions.
In NLP, classifier transformer architectures—introduced with \gls*{bert}~\cite{devlin2019bert}—redefined classification via self-attention, and have since been extended to multi-modal data. Analogously, group trajectories can be viewed as structured “sentences” generated by grammar rules, with dependencies forming syntax trees. This motivates using \gls*{gtnn}, which propagates information across parse trees and captures hierarchical dependencies beyond sequence models.

Nodes in the GTNN correspond to grammar production rules, edges encode generation dependencies, and graph-based self-attention highlights influential coalitions and long-range interactions. This extends transformers from token sequences to structured trajectories, enabling principled inference of group intent.

To summarize, we construct a multi-layer modeling and algorithmic framework  between the target tracker (Bayesian signal processing level) and the self-attention mechanism of the \gls*{gtnn} (group intent inference).

\subsection{Related Work}

\subsubsection*{Stochastic Grammar based Trajectory Modelling}

\gls*{nlp} models and associated statistical signal processing algorithms have been previously used in trajectory analysis. Stochastic Context-Free Grammars and Reciprocal Process Models were used in \cite{FK13}, and were extended to more complex metalevel tracking scenarios in \cite{FK14, FK15}. Syntactic tracking using GMTI measurements is studied in  \cite{WKB11, krishnamurthy2018syntactic}. Recent work has also studied embedded stochastic syntactic processes that are equivalent to Markov processes, highlighting new opportunities for grammar-based trajectory inference~\cite{ESSP21}. However, stochastic context-free grammars, while more general than Hidden Markov Models (HMM), are less expressive than \gls*{scsg}. The potential of \gls*{scsg}-based models and associated statistical signal processing for inference of metalevel tracking remains largely unexplored and has been addressed in this paper.

\subsubsection*{Cooperative Game Theory for Group Targets}

Cooperative game theory provides mathematical tools for modeling how groups of agents form coalitions and share payoffs. Central concepts include the \emph{core}, which captures allocations where no coalition has incentive to deviate \cite{peleg2007introduction,osborne1994course}, and the \emph{Shapley value}, which fairly distributes payoffs based on each agent’s average marginal contribution \cite{shapley1953value}. The \emph{nucleolus} uniquely minimizes dissatisfaction among coalitions and lies in the core whenever it is non-empty \cite{schmeidler1969nucleolus}. These solution concepts have been applied in diverse domains such as economics, network design, and multi-agent systems, offering principled ways to model fairness, stability, and cooperation. In aeronautical systems, cooperative-game formulations have been used for UAV handoff decisions and distributed resource/scheduling problems \cite{CG1, CG2}.
In control theory, the notion of \emph{group intent} refers to the collective objectives or emergent behaviors that arise when multiple agents interact locally without centralized coordination, as studied in swarm intelligence models such as Reynolds’ Boids and Vicsek’s flocking \cite{reynolds1987flocks}. While these models emphasize emergent patterns, control-theoretic approaches such as soft control \cite{han2006soft} have sought to deliberately shape group outcomes. 
In contrast, our work utilizes the characteristic function as a fair allocation yielding the probability distribution of the stochastic grammar based generative model for trajectory generation. 

\subsubsection*{Deep Learning Approaches for Trajectory Intent Inference} 
Trajectory classification and intent recognition play a central role in applications such as autonomous driving, surveillance, and airspace monitoring. 
Traditional approaches often rely on statistical models like HMMs or feature-based classifiers~\cite{Kandeepan14, HMM_traj, ParalleleA}, which are limited in capturing long-range or structured behavior. 
Bayesian inference techniques have also been applied to jointly estimate states and predict intent, enabling more robust handling of uncertainty in human and object trajectories~\cite{liang2020intent, Bashar1}.
More recent methods have utilized deep learning to infer agent intent from observed motion~\cite{choi2021drogon , teo2024automatic}. 
In our work, we leverage grammar-aware inference to exploit the structural knowledge encoded in T-structured data. Tree-structured neural networks, such as Tree-LSTMs~\cite{tai2015improved} and recurrent neural network grammars (RNNGs)~\cite{dyer2016recurrent}, have been successfully applied to syntactic parsing and semantic representation, demonstrating stronger generalization compared to flat sequence models. Similarly, \gls*{gnn} have emerged as powerful methods for learning over hierarchical data structures, making them well-suited for modeling parse trees and structured trajectories~\cite{kipf2017semi}. Building on these advances, our approach integrates a grammar-aware neural network to infer the characteristic function, which we interpret as the intent of the group target.

\subsection{Organization and Main Results}
To capture how cooperative game among group targets, we connect the game-theoretic representation of group intent with the sequential dynamics of each target. The overall modeling pipeline is summarized in Equation~\eqref{eq:master}.
Equation~\eqref{eq:master} presents the high-level formulation linking the \emph{group intent} to the targets’ velocity. 
\begin{equation}
\label{eq:master}
\; v 
\;\xrightarrow[\eqref{eq:Nucleolus}]{\mathcal{N}(u)} 
\; \pi^\ast 
\;\xrightarrow[\eqref{eq:norm_prob}]{\mathcal{P}(r)} 
\; \Gamma
\;\xrightarrow[\eqref{eq:statespace}]{\text{trajectory generation}}
\{\vec{v}_k\}
\end{equation}
Here $\utility$ denotes the characteristic function of the cooperative game representing the group intent. The probabilities \( \mathcal{P}(r) \) in the production rule \( \Gamma \)  are derived via the nucleolus allocation \( \pi^\ast\) for each target. The grammar rules then determine the velocity sequence \( \{\vec{v}_k\} \) for all time steps \( k \) which influences the trajectories of the targets. 
During inference, the process is reversed: starting from state estimates \(\{\hat{\vec{v}}_k\}\), we construct a grammar parse tree \(T\), which is then processed by a \gls*{gtnn} to infer the underlying group intent \(\hat{u}\). The organization of this paper is as follows:

\begin{enumerate}
    \item In Sec.~\ref{sec:grammar}, we use stochastic formal grammar to model and parse group trajectories, capturing dependencies beyond traditional Markov chains. 
    \item In Sec.~\ref{sec:game}, we connect cooperative game theory with stochastic formal-grammar in Theorem~\ref{theorem:zero_payoff}, where production-rule probabilities of the grammar serve as allocation vectors within the cooperative game’s core. We further propose a Fisher–information–based characteristic function to characterize this core and prove in Theorem~\ref{theorem:trace_supermodular} that it is modular. The modularity ensuring fair and efficient allocation of trajectory sub-tasks since the Shapley value lies in the core. This formulation provides a principled framework for modeling and inferring group intent.

    \item In Sec.~\ref{sec:classifier}, we examine inference for \DIFdelbegin \DIFdel{SRG and SCFG}\DIFdelend \DIFaddbegin \DIFadd{stochasic regular grammar (SRG) and stochastic context free grammar (SCFG)}\DIFaddend , highlighting their efficiency and limitations, which motivates the use of SCSG. Trajectories are mapped to context-sensitive languages and parsed into grammar trees, which serve as input to a \gls*{gtnn} with graph self-attention, pooling, and dense layers. This design enables end-to-end inference of the characteristic function—interpreted as group intent—directly from trajectories, while leveraging the structural knowledge encoded in the grammar. Our approach thus integrates tracker-level data with graph-based self-attention for group intent inference.
    \item In Sec.~\ref{sec:experiments}, we compare the proposed Fisher-information–based characteristic function with other baselines, demonstrating improvements for better group utility allocation.
    We also evaluate the grammar-aware graph neural network against other baselines, showing improvements over methods that ignore the structural knowledge encoded in the stochastic formal grammar.

\end{enumerate}

\section{Background: Kinematic and Trajectory Models}
\label{sec:grammar}

In this section, we introduce a Bayesian trajectory framework that integrates conventional state-space models for target dynamics with grammar-based representations of motion. The state-space formulation captures the evolution of target kinematics and their noisy observations, while the grammar provides a structured, stochastic mechanism for composing low-level velocities into higher-order geometric and spatiotemporal patterns.
It is important to emphasize that the trajectory models and resulting algorithms discussed in this paper operate seamlessly with the classical target tracking algorithms.

In Sec.~\ref{sec:grammar}–\ref{subsec:grammar_dynamics}, we present a Bayesian trajectory framework that embeds conventional state-space models for target tracking within a grammar-based representation. The grammar composes geometric primitives into high-level spatio-temporal patterns to capture complex motion trajectories.
In Sec.~\ref{sec:grammar}–\ref{subsec:IIB}, we present the proposed metalevel tracking framework, detailing how collective motion trajectories are parsed and structured for integration into a grammar-based inference model that enables higher-level group behavior analysis.
In Sec.~\ref{sec:grammar}–\ref{subsec:grammar_modeling}, we introduce the foundations of formal grammar theory and analyze the expressive capabilities of different grammar classes for intent modeling. This section constitutes the background for Sec.~\ref{sec:game} where we will formulate {\em group intent}  as a cooperative game, and the outcome of the game  will modulate the probabilities of the grammar representation and thus the trajectories.

\subsection{Target Dynamics and Observation Model}
\label{subsec:grammar_dynamics}
In classical target tracking radars\cite{bar2004estimation,blackman1999design},  the  kinematic state of a target (position  and velocity) at time \( k \) is represented by \( \state_k =[p_k^{1}, p_k^{2}, v_{k}^{1}, v_{k}^{2} ]^{\top} \) where $p_{k}^{i}$ is the target's position in the $i^{\text{th}}$ dimension at time $k$, and its observation is given by \( \observation_k \in \mathbb{R}^4 \). We assume that the state specifies the target's position and velocity in a 2-dimensional space. The state evolves as 
\begin{equation}
\label{eq:statespace}
\state_{k+1} = \mathbf{F} \state_k + \mathbf{G} \control_k+ \statenoise_k,
\end{equation}
where $\mathbf{F}$ and $\mathbf{G}$ matrices are defined in \cite[Ch.\,2.6]{Krishnamurthy_2016}. The \emph{i.i.d.} process noise  is denoted as \( \statenoise_k \sim\mathcal{N}(0,\mathbf{Q}) \), $\mathbf{Q}$ defined in \cite[Ch.\,2.6]{Krishnamurthy_2016}. 
Although the acceleration vector $\control_k$ ultimately determines the trajectory, in our
syntactic layer we encode the trajectory via the velocity
sequence $v_k=[\dot{p}_{k}^{1}, \dot{p}_{k}^{2}]^{\top}$, where $\dot{p}_{k}^{i}$ is the velocity in the $i^{\text{th}}$ dimension at time $k$. The velocity sequence $v_k$ is determined by the production rules defined in the next subsection and summarized in \eqref{eq:master}. 
The observation at time \( k \) recorded by the radar  is 
\begin{equation}
\label{eq:observe}
\observation_k = \state_k+ \observationNoise_k,
\end{equation}
where \( \observationNoise_k \sim \mathcal{N}(0,\boldsymbol{\Sigma}) \) is the \emph{i.i.d.} measurement noise \cite[Ch.\,2.6]{Krishnamurthy_2016}.\\ 
{\bf Target Tracker}.
A Bayesian tracker computes a sequence of (possibly approximate) posterior distributions $\{\Pi_k\}$ as 
\begin{equation}
\label{eq:bayes}
\Pi_{k+1} = \mathcal{B}(\Pi_k, \observation_k).
\end{equation}
The Bayesian recursion in~\eqref{eq:bayes} updates the posterior $\Pi_k$ over the target state using the incoming observation $\observation_k$, thereby integrating prior information, process dynamics, and measurement likelihoods to produce the updated posterior $\Pi_{k+1}$. 
In practice, $\mathcal{B}$ may represent either the optimal Bayesian filter or an approximate inference scheme such as a particle filter or interacting multiple model (IMM) algorithm~\cite{doucet2001sequential, IMM}. 
The state estimate at time $k+1$ is then extracted from the posterior through a suitable estimator function as

\begin{equation}
\label{eq:track_est}
\stateinf_{k+1} = \Phi(\Pi_{k+1}),
\end{equation}
where \( \Phi(\cdot) \) denotes the state estimation function that yields the most likely state from the updated posterior. Following a similar Bayesian update process, the multi-target tracker estimates the joint state $\hat{\mathbf{X}}_k^{N}$ comprising the individual target estimates $\hat{\state}_k^{n}$ for $n = 1, \ldots, N$. 
In the multi-target case with \( N \) targets, the overall system state at time \( k \) is represented as
\begin{equation}
\label{eq:MultiTarget}
    \mathbf{\hat X}_k^{N} = \{[\hat{{p}}_{1k}^{1}, \hat{{p}}_{1k}^{2}], \ldots, [\hat{{p}}_{Nk}^{1}, \hat{{p}}_{Nk}^{2}],[\hat{{v}}_{1k}^{1}, \hat{{v}}_{1k}^{2}], \ldots, [\hat{{v}}_{Nk}^{1},\hat{{v}}_{Nk}^{2}]\},
\end{equation}
where $\hat{p}_{jk}^{i}$ and $\hat{v}_{jk}^{i}$ are the positions and velocities of the $j=1,2,\cdots, N^{\mathrm{th}}$ target, in the $i^{\mathrm{th}}$ dimension at $k=0,1,2,\cdots$ time step.
Here, $\mathbf{\hat X}_k^{N}$ represents the unlabeled finite set of $N$ target states at time $k$, such that the ordering of elements is immaterial and no data association is imposed. By design of the meta-level tracker to infer the group intent defined in the next section no data-association is required. 


\begin{figure}[h!]
    \centering
    \includegraphics[width=0.5\linewidth]{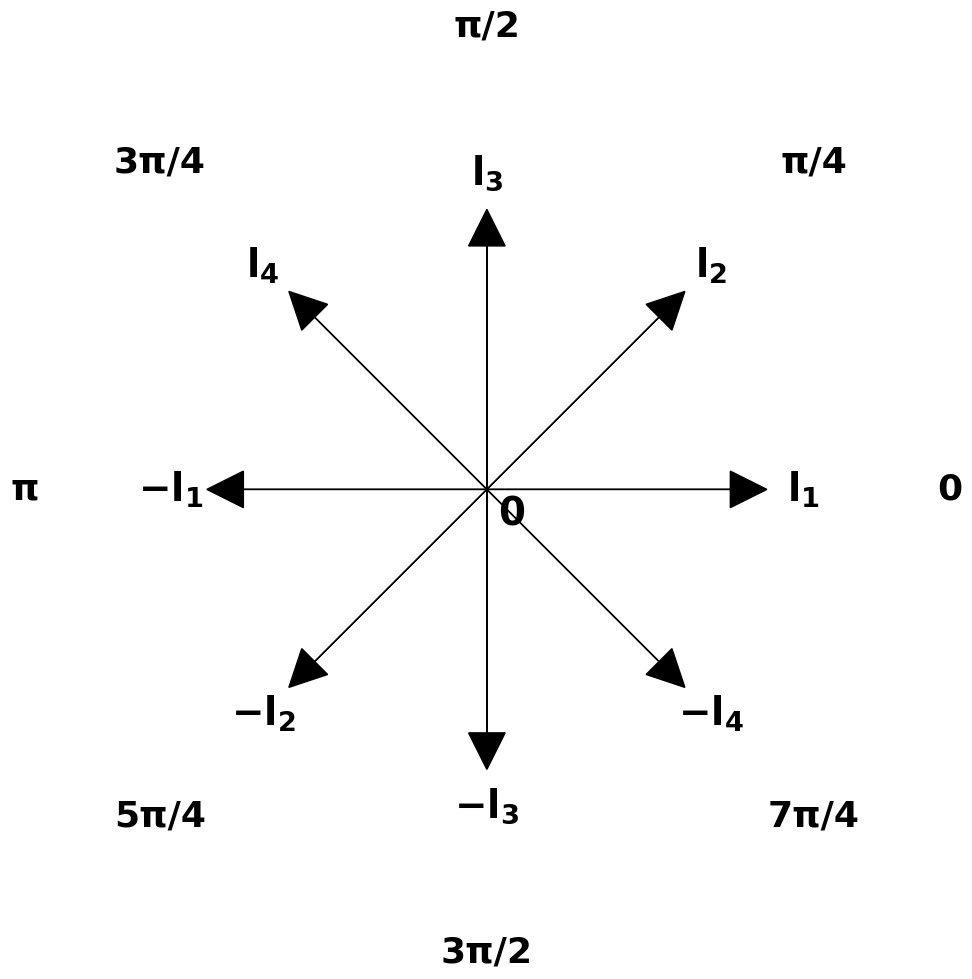}
    \caption{Illustration of the normalized velocity vectors
    $[v_{k}^{1},v_{k}^{2}] \in \{0, l_1, l_2, l_3, l_4, -l_1, -l_2, -l_3, -l_4 \}$  used in the kinematic description~\eqref{eq:statespace}.}
    \label{fig:velocity}
\end{figure}

\subsection{Metalevel Tracker}
\label{subsec:IIB}

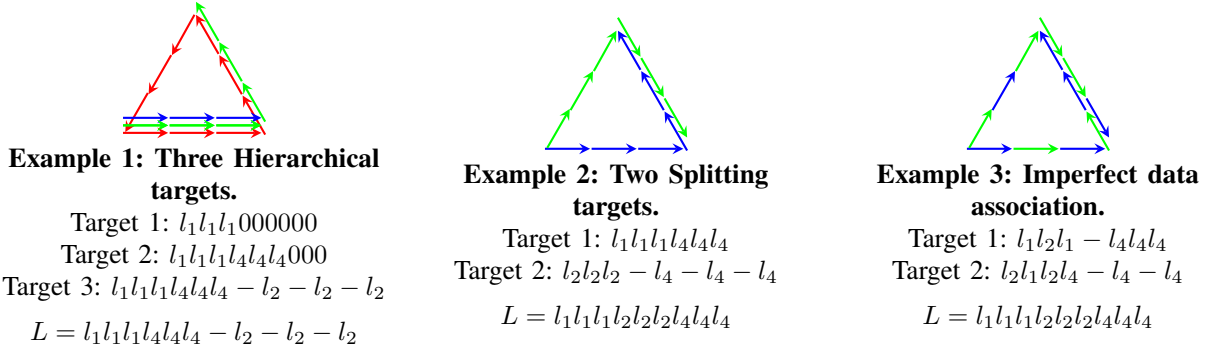
\begin{figure*}[t!]
\centering
\begin{tabular}{ccc}

\begin{minipage}{0.3\textwidth}
\centering
\begin{tikzpicture}[>=stealth, thick, scale=1.0]

\draw[->, red] (0.2,0) -- (0.8,0);
\draw[->, red] (0.82, 0) -- (1.42,0);
\draw[->, red] (1.44, 0) -- (2.04,0);

\draw[->, red] (2.09,0.15-0.173) -- (1.79 ,0.67-0.173);
\draw[->, red] (1.78 ,0.687-0.173) -- (1.48 ,1.207-0.173);
\draw[->, red] (1.47, 1.207+0.017-0.173) -- (1.17 ,1.207+0.017 + 0.52-0.173);

\draw[->, red] (1.17 -0.02 , 1.207+0.017 + 0.52-0.173) -- (1.17 -0.32  ,1.207+0.017 + 0.52-0.173 - 0.31*1.732);
\draw[->, red] (1.17 -0.32 - 0.01  , 1.207+0.017 + 0.52-0.173 - 0.31*1.732) -- (1.17 -0.32 - 0.01 - 0.3 ,1.207+0.017 + 0.52-0.173 - 0.31*1.732 - 0.52);
\draw[->, red] (1.17 -0.32 - 0.02 - 0.3 ,1.207+0.017*2 + 0.52-0.173 - 0.31*1.732 - 0.52 ) -- (1.17 -0.32 - 0.02 - 0.3*2 ,1.207+0.017*2 + 0.52-0.173 - 0.31*1.732 - 0.52*2);

\draw[->, blue] (0.2,0.2) -- (0.8,0.2);
\draw[->, blue] (0.82, 0.2) -- (1.42,0.2);
\draw[->, blue] (1.44, 0.2) -- (2.04,0.2);

\draw[->, green] (0.2,0.1) -- (0.8,0.1);
\draw[->, green] (0.82, 0.1) -- (1.42,0.1);
\draw[->, green] (1.44, 0.1) -- (2.04,0.1);

\draw[->, green] (0.2,0.1) -- (0.8,0.1);
\draw[->, green] (0.82, 0.1) -- (1.42,0.1);
\draw[->, green] (1.44, 0.1) -- (2.04,0.1);

\draw[->, green] (2.09,0.15) -- (1.79 ,0.67);
\draw[->, green] (1.78 ,0.687) -- (1.48 ,1.207);
\draw[->, green] (1.47, 1.207+0.017) -- (1.17 ,1.207+0.017 + 0.52);

\end{tikzpicture}

\textbf{Example 1: Three Hierarchical targets.}\\
Target 1: $l_1 l_1 l_1 000000$\\
Target 2: $l_1 l_1 l_1 l_4 l_4 l_4 000$\\
Target 3: $l_1 l_1 l_1 l_4 l_4 l_4 -l_2 -l_2 -l_2$
\[
L = l_1 l_1 l_1 l_4 l_4 l_4 -l_2 -l_2 -l_2
\]
\end{minipage}
&
\begin{minipage}{0.3\textwidth}
\centering
\begin{tikzpicture}[>=stealth, thick, scale=1.0]

\draw[->, blue] (0.2+2.5,0) -- (0.8+2.5,0);
\draw[->, blue] (0.82+2.5, 0) -- (1.42+2.5,0);
\draw[->, blue] (1.44+2.5, 0) -- (2.04+2.5,0);

\draw[->, blue] (2.09+2.5,0.15-0.173) -- (1.79+2.5 ,0.67-0.173);
\draw[->, blue] (1.78+2.5 ,0.687-0.173) -- (1.48+2.5 ,1.207-0.173);
\draw[->, blue] (1.47+2.5, 1.207+0.017-0.173) -- (1.17+2.5 ,1.207+0.017 + 0.52-0.173);

\draw[->, green]   (1.17 -0.32+2.5  ,1.207+0.017 + 0.52-0.173 - 0.31*1.732 ) -- (1.17 -0.02+2.5 , 1.207+0.017 + 0.52-0.173);
\draw[->, green]  (1.17 -0.32 - 0.01 - 0.3 +2.5 ,1.207+0.017 + 0.52-0.173 - 0.31*1.732 - 0.52) -- (1.17 -0.32 - 0.01+2.5  , 1.207+0.017 + 0.52-0.173 - 0.31*1.732);
\draw[->, green] (1.17 -0.32 - 0.02 - 0.3*2 +2.5 ,1.207+0.017*2 + 0.52-0.173 - 0.31*1.732 - 0.52*2) -- (1.17 -0.32 - 0.02 - 0.3 +2.5 ,1.207+0.017*2 + 0.52-0.173 - 0.31*1.732 - 0.52 );

\draw[->, green]   (1.79+2.5 ,0.67 ) -- (2.09+2.5,0.15);
\draw[->, green] (1.48+2.5 ,1.207)  --  (1.78+2.5 ,0.687);
\draw[->, green]  (1.17+2.5 ,1.207+0.017 + 0.52)  -- (1.47+2.5, 1.207+0.017);

\end{tikzpicture}

\textbf{Example 2: Two Splitting targets.}\\
Target 1: $l_1 l_1 l_1 l_4 l_4 l_4$\\
Target 2: $l_2 l_2 l_2 -l_4 -l_4 -l_4$
\[
L = l_1 l_1 l_1 l_2 l_2 l_2 l_4 l_4 l_4
\]
\end{minipage}
&
\begin{minipage}{0.3\textwidth}
\centering
\begin{tikzpicture}[>=stealth, thick, scale=1.0]

\draw[->, blue] (0.2+5,0) -- (0.8+5,0);
\draw[->, green] (0.82+5, 0) -- (1.42+5,0);
\draw[->, blue] (1.44+5, 0) -- (2.04+5,0);

\draw[->, green] (2.09+5,0.15-0.173) -- (1.79+5 ,0.67-0.173);
\draw[->, blue] (1.78+5 ,0.687-0.173) -- (1.48+5 ,1.207-0.173);
\draw[->, blue] (1.47+5, 1.207+0.017-0.173) -- (1.17+5 ,1.207+0.017 + 0.52-0.173);

\draw[->, green]   (1.17 -0.32+5  ,1.207+0.017 + 0.52-0.173 - 0.31*1.732 ) -- (1.17 -0.02+5 , 1.207+0.017 + 0.52-0.173);
\draw[->, blue]  (1.17 -0.32 - 0.01 - 0.3 +5 ,1.207+0.017 + 0.52-0.173 - 0.31*1.732 - 0.52) -- (1.17 -0.32 - 0.01+5  , 1.207+0.017 + 0.52-0.173 - 0.31*1.732);
\draw[->, green] (1.17 -0.32 - 0.02 - 0.3*2 +5 ,1.207+0.017*2 + 0.52-0.173 - 0.31*1.732 - 0.52*2) -- (1.17 -0.32 - 0.02 - 0.3 +5 ,1.207+0.017*2 + 0.52-0.173 - 0.31*1.732 - 0.52 );

\draw[->, blue]   (1.79+5 ,0.67 ) -- (2.09+5,0.15);
\draw[->, green] (1.48+5 ,1.207)  --  (1.78+5 ,0.687);
\draw[->, green]  (1.17+5 ,1.207+0.017 + 0.52)  -- (1.47+5, 1.207+0.017);

\end{tikzpicture}

\textbf{Example 3: Imperfect data association.}\\
Target 1: $l_1 l_2 l_1 -l_4 l_4 l_4$\\
Target 2: $l_2 l_1 l_2 l_4 -l_4 -l_4$
\[
L = l_1 l_1 l_1 l_2 l_2 l_2 l_4 l_4 l_4
\]
\end{minipage}

\end{tabular}

\caption{Three examples showing hierarchical, splitting, and imperfect data-association behaviors, visualized via arrow-based motion diagrams and their resulting grammar transformations.}
\label{fig:three_examples}
\end{figure*}

Based on the estimated state sequences $\mathbf{\hat X}_{k}^{N}$ in \eqref{eq:MultiTarget}, we introduce a \emph{metalevel tracking} framework that operates above the traditional kinematic layer. Unlike conventional trackers that estimate individual target trajectories, the metalevel tracker analyzes the collective motion behavior of multiple targets to infer higher-level group intent. 

To achieve this, we employ a parsing procedure that transforms the collective velocity trajectories of all targets into a single symbolic representation. This symbolic sequence serves as the input to a formal grammar-based inference mechanism for group-level behavior analysis. Each velocity vector $[\hat{{v}}_{jk}^{1},\hat{{v}}_{jk}^{2}]$ in \eqref{eq:MultiTarget} is quantized according to Fig.~\ref{fig:velocity}, relative to its nearest representative vector in the predefined quantization set. The resulting symbolic representation is then obtained through the following parsing procedure:

\begin{itemize}
    \item Velocities corresponding to spatially overlapping positions of different targets are ignored.
    \item Zero velocities, i.e., $[\hat{{v}}_{jk}^{1},\hat{{v}}_{jk}^{2}] = [0,0]$, are excluded from further processing.
    \item Consecutive velocity vectors with the same or opposite directions are treated as part of the same motion pattern.
    \item Distinct velocity directions observed over time are recorded sequentially to form the symbolic representation of collective motion.
\end{itemize}

To illustrate how the above procedure operates in practice, we provide a set of representative examples, as shown in Fig.~\ref{fig:three_examples}. These examples demonstrate how multiple individual target motion sequences are progressively merged into a unified symbolic representation. The transformation highlights the model’s ability to capture the collective dynamics of a group, eliminate redundancies, and preserve the temporal and spatial coherence of motion patterns across targets.

These examples demonstrate the core principle of the metalevel tracking framework: group target trajectories, once parsed and symbolically transformed, can be collectively represented through a unified sequence. This unified representation serves as the foundation for the next stage of analysis, where motion patterns are encoded using a stochastic formal grammar to infer group-level intent and coordinated behaviors. The following section describes this grammar-based inference process in detail.



\begin{figure*}[h!]
\centering
\resizebox{0.97\textwidth}{!}{%
\begin{tikzpicture}[
    node distance=1.1cm and 1.1 cm,
    box/.style={rectangle, rounded corners=3pt, thick, draw, minimum width=3cm, minimum height=1cm, align=center},
    ->, >=Stealth
]

\node[box] (u1) {Group Target Intent  \\ as Characteristic Function};
\node[above=1.1cm of u1.west, anchor=west] (subtitle1) {\normalsize \textbf{Group Trajectory Generative Model Formulation}};

\node[box, right=of u1] (n1) {Payoff Allocation of \\ Cooperative Game};
\node[box, right=of n1] (prob1) {Production Rule \\ Probabilities of  \\ Stochastic Formal Grammar};
\node[box, right=of prob1] (traj1) {Coordinated  \\ Group Trajectory};

\node[box, below=1.5cm of traj1] (traj2) {Observed \\ Group Trajectory};
\node[box, left=2.9cm of traj2] (prob2) {Grammar Parse Tree \\ of Stochastic Formal Grammar};
\node[box, left=2.95cm of prob2] (n2) {Inferred Group Intent \\ as Characteristic Function \\ of the Game};
\node[above=1.1cm of n2.west, anchor=west] (subtitle2) {\normalsize \textbf{Group Intent Inference Formulation}};

\draw[->] (u1) -- (n1);
\draw[->] (n1) -- (prob1);
\draw[->] (prob1) -- (traj1);

\draw[->] (traj2) -- (prob2);
\draw[->] (prob2) -- (n2);

\draw[dashed, ->] (traj1) -- (traj2);

\end{tikzpicture}%
}

\caption{Overview of the proposed framework. The top line shows the group trajectory generative model formulation: Group intent is modeled as the characteristic function of a cooperative game; allocations based on characteristic function induce probabilities over production rules in a stochastic formal grammar, which generate coordinated group trajectories. The bottom line shows the group intent inference formulation: From observed group trajectories, a grammar parse tree is constructed and used to recover the group intent by estimating the characteristic function.}
\label{fig:idea}
\end{figure*}
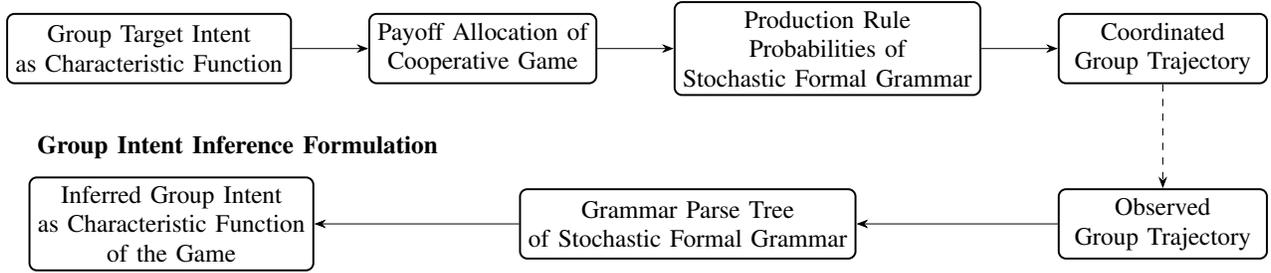


\subsection{Geometric Shape-Based Intent Modeling Using Stochastic Grammars}
\label{subsec:grammar_modeling}
This section demonstrates how the intent of a group of targets can be modeled by framing the estimated track sequence $\{\hat{\vec{v}}_k\}:=\{[v_{jk}^{1},v_{jk}^{2}]\}$ in \eqref{eq:MultiTarget} into geometric shapes. We focus on a higher level of abstraction, involving meta-level tracking. The key idea is to use a Stochastic Grammar from Formal Language Theory to fit geometric shapes as foundational elements for trajectory modeling to the sequence of track estimates  $\{\hat{\vec{v}}_k\}$  generated above. 

\noindent{\bf Stochastic Grammars}:  
A stochastic grammar is defined as $G = (\mathcal{E}, \mathcal{A}, \Gamma, \mathcal{P})$, where $\mathcal{E}$ is the set of non-terminal symbols, $\mathcal{A}$ is the set of terminal symbols, $\Gamma$ is the set of production rules, and $\mathcal{P}$ is a probabilistic map defining the distribution over the production rules $\Gamma$. These grammars are stochastic because each non-terminal has multiple production rules, with the selection made randomly. As summarized in \eqref{eq:master}, the $\Gamma$ influences the velocity $\{\vec{v}_k\}$ in \eqref{eq:statespace}. 
Specifically, each terminal in $\mathcal{A}$ represents a directional velocity in the trajectory. 
For example, as shown in Fig.~\ref{fig:velocity}, each vector represents a directional velocity governed by the velocity term ${\vec{v}_k}$, with the set of terminal symbols defined as $\mathcal{A} = \{0, l_1, l_2, l_3, l_4, -l_1, -l_2, -l_3, -l_4\}$. In this way, each velocity can be mapped one-to-one to a terminal in the formal grammar $G$, constructing the bridge between target dynamics and stochastic formal grammar theory.

\noindent{\bf Generative Models for Trajectories}. By a generative model for a trajectory, we mean a grammar that can exclusively generate these  trajectories. This is typically verified using pumping lemmas \cite{Pumping}. Different types of grammars are suited to modeling different classes of trajectories:

\noindent{\em Stochastic Regular Grammars (SRGs)}, which are equivalent to finite-state Markov chains or Hidden Markov Models, are efficient for modeling linear trajectories but are fundamentally limited by the first-order Markov assumption. This means they cannot capture long-range dependencies, enforce loop closures, or represent branching decisions, making them suitable only for simple, sequential paths without hierarchy.
\\
{\em Stochastic Context-Free Grammars (SCFGs)} introduce hierarchical structure, allowing recursive and nested relationships that go beyond the flat, sequential patterns of SRGs. This makes it possible to capture long-range dependencies, represent repeated substructures. As a result, SCFGs can express more complex movement patterns such as m-rectangles or arcs while preserving the probabilistic framework needed for uncertainty in trajectory modeling. \\
{\em Stochastic Context-Sensitive Grammars (SCSGs)} provide the power to model context-dependent trajectories, closed-loop behaviors such as triangles, rectangles, or squares by encoding dependencies across distant segments. However, Bayesian inference in such grammar is NP-hard, making them computationally prohibitive; therefore, we instead employ deep neural networks to approximate these capabilities in a tractable manner.

By selecting the appropriate generative grammar, we can effectively model a broad spectrum of trajectory behaviors, each with varying levels of structural and contextual complexity. 
We utilize grammar $G$ for trajectory modeling and analysis. Let $\alphabet $ be a finite set of terminals, where each element within represents a unit velocity of the target in a specific  direction. Recall that the velocity $\vec{v}_k$ defined in \eqref{eq:statespace}. We define $\mathcal{E}$ to be a finite set of non-terminals, where each element within represents an intermediate state of the target. For further details on the production rules $\Gamma$ for SRGs, SCFGs, and SCSGs refer to \cite{Fu:1982:SyntactPatternRecognition}. Given a starting non-terminal in $\mathcal{E}$, guided by stochastic production rules $\Gamma$ and distribution over production rules $\mathcal{P}$, a string $\{l_i\}_{i=1}^{T}$ of grammar $G$ can be generated. 

In Sec.\ref{sec:experiments}, to model realistic deviations arising from both the target dynamics and the observation process in \eqref{eq:statespace} and~\eqref{eq:observe}, we incorporate noise directly into the grammar 
$G$. Specifically, we extend the production mechanism by introducing a noise terminal within the rule set, allowing random perturbations to emerge during generation (see \eqref{eq:Perturbation}). This extension yields a grammar that captures not only structured trajectory patterns but also the noisy behaviors commonly encountered in tracking scenarios.

\section{Modeling Group Intent as Allocations in a Cooperative Game}
\label{sec:game}

We are now ready to present our first main result. 
We present our formulation of group intent as the characteristic function of a cooperative game played by a group of targets. The overview of our framework is shown in Fig.~\ref{fig:idea}. The goal is to bridge low-level trajectory tracking with high-level intent inference by integrating Bayesian state-space modeling, stochastic formal grammars, and cooperative game theory.

Inferring group intent requires more than analyzing individuals in isolation. 
While non-cooperative or purely statistical approaches can model independent behaviors or correlations, they struggle to explain how group targets coordinate toward shared objectives. 
A cooperative game–theoretic formulation is particularly well-suited for this challenge: it explicitly captures coalition formation, the creation of joint value, and principled allocation rules that ensure fairness or stability. 
This provides a rigorous connection between motion-level trajectories and high-level intent, making cooperative games as a particularly suitable framework.

In Sec.~\ref{sec:game}-\ref{subsec:cooperative_game}, we model group intent via the \emph{characteristic function} of a cooperative game~\eqref{eq:v}, with the allocation~\eqref{eq:Nucleolus} prescribing sub-task distribution among targets. This allocation directly determines production rule probabilities in the stochastic grammar~\eqref{eq:norm_prob}, embedding both motion syntax and rational coordination. Theorem~\ref{theorem:zero_payoff} proves the formulation’s effectiveness.
To make this connection concrete, we focus on three canonical cooperative-game allocations:
\begin{enumerate}
    \item \textbf{Core:} Ensures that no coalition of agents has an incentive to deviate, guaranteeing stability of the allocation.
    \item \textbf{Nucleolus:} Minimizes coalition dissatisfaction by lexicographically minimizing excess payoffs, yielding a balanced and robust allocation.
    \item \textbf{Shapley Value:} Provides an equitable distribution of the group’s total value by averaging marginal contributions across all coalition orderings.
\end{enumerate}

In Sec.~\ref{sec:game}–\ref{subsec:fisher}, as a concrete example, we introduce a Fisher-information–based characteristic function \eqref{eq:fisher} and demonstrate its effectiveness in intent modeling using Theorem~\ref{theorem:zero_payoff}.

\subsection{Characteristic Function of the Cooperative Game as the Intent of the Trajectory}
\label{subsec:cooperative_game}


Given a group of $N$ players (total number of targets), we consider a joint coalition vector $S \in \{0,1\}^N$, where $S_i = 1$ indicates that player $i$ participates the game (i.e., contributes a sub-trajectory to the joint group trajectory), and $S_i = 0$ otherwise. The configuration $S$ thus specifies a coalition of the game. 
We model the cooperative game with a characteristic function:
\begin{equation}
\label{eq:v}
    \utility: 2^N \to \mathbb{R},
\end{equation}
where \( \utility(S) \) quantifies the achievable utility when only the members of coalition $S$ contribute their trajectories.
Given characteristic function $\utility(S)$, the \emph{core} of the game is defined as
\DIFaddbegin {\small\DIFaddend \begin{equation}
\label{eq:Core}
C(\utility) = \left\{ \pi \in \mathbb{R}_{+}^{N} \,\middle|\, \sum_{i \in N} \pi_i = u(N), \sum_{i \in S} \pi_i \geq \utility(S), \forall S \subseteq N \right\}.
\end{equation}\DIFaddbegin }
\DIFaddend The core provides a set of stable allocations where no coalition has an incentive to deviate and form on its own. 
The next question is how can we choose one single point within the core $C(\utility)$. 
One important solution is the concept of \emph{nucleolus}, which is an optimal point within the core proved in~\cite{schmeidler1969nucleolus}. 
For an allocation $\pi \in \mathbb{R}^N$ and coalition $S \subseteq N$, the \emph{excess} of $S$ with respect to $\pi$ is defined as
\begin{equation}
    e(S,\pi) := \utility(S) - \sum_{i \in S} \pi_i.
\end{equation}
The \emph{vector of excesses} under allocation $\pi$ is given by
\begin{equation}
    \theta(\pi) = \big( e(S_1,\pi), e(S_2,\pi), \dots, e(S_{2^N},\pi) \big),
\end{equation}
where the excesses are arranged in non-increasing order (i.e.,
$
    e(S_m,\pi) \geq e(S_n,\pi) \quad \text{for } m \leq n
$).
Finally, for a cooperative game $\mathcal{CG}(N,u)$, the \emph{nucleolus} is defined as
\begin{equation}
\label{eq:Nucleolus}
\mathcal{N}(\utility) = \left\{ \pi^\ast \in \mathbb{R}_{+}^{N} \,\middle|\, \rho(\pi^\ast) \preceq_{\text{lex}} \rho(\pi), \forall \pi \in \mathbb{R}_{+}^{N} \right\},
\end{equation}
where $\preceq_{\text{lex}}$ denotes the lexicographic order: two vectors $\rho(\pi^\ast)$ and $\rho(\pi)$ are compared by looking at their first components; if they are equal, the comparison moves to the second component, and so on.
Given a cooperative game $\mathcal{CG}(N,u)$, the nucleolus $\mathcal{N}(u)$ is a single point within the core so that the payoff allocation $\pi^\ast$ in which total utility $v(N)$ is fully distributed among players, and no coalition $S$ can obtain a higher collective payoff by breaking away from the grand coalition. 

Another important point solution concept in cooperative game theory is the \emph{Shapley value}~\cite{shapley1953value}, which provides a fair allocation of the total utility based on each player’s marginal contribution across all possible coalitions. For a cooperative game $\mathcal{CG}(N,u)$, the Shapley value is defined as
\begin{align}
\label{eq:Shapley}
\phi(u) = \bigg\{ \pi^* \in \mathbb{R}_{+}^{N} \,\bigg|\,
\pi^*_i &= \sum_{S \subseteq N \setminus \{i\}} 
\frac{|S|!(|N|-|S|-1)!}{|N|!} \nonumber \\
&\quad \cdot \big( \utility(S \cup \{i\}) - \utility(S) \big), \,
\forall i \in N \bigg\}.
\end{align}
The Shapley value assigns to each player $i$ their expected marginal contribution when they join coalitions in all possible orders, ensuring a fair and symmetric division of the total payoff. Moreover, in convex cooperative games\cite{shapley1953value}, the Shapley value always lies within the core, further reinforcing its stability as a solution concept.
This formulation not only ensures theoretical rigor but also offers practical utility for real-world multi-target systems where stability and fairness are paramount.

Based on the allocation obtained from \emph{core}, \emph{nucleolus}, or \emph{Shapley value}, we construct a \emph{probabilistic map} $\mathcal{P}$ that assigns probabilities to the production rules $\Gamma$. For each production rule $r$, let $I_r$ be the set of targets assigned to that rule. The numerator of $\mathcal{P}(r)$ is the sum of the payoffs $\pi_i$ for all $i \in I_r$, representing the contribution allocated to rule $i$. To normalize, we define $\mathcal{Q}(r)$ as the set of all production rules sharing the same left-hand side as $i$. For each rule $r \in \mathcal{Q}(r )$, we sum the payoffs of its assigned targets, and then sum across all such rules. The probability is therefore
\begin{equation}
\label{eq:norm_prob}
    \mathcal{P}(i) = \frac{\sum_{j \in I_i} \pi^\ast_j}{\sum_{k \in \mathcal{Q}(i)} \sum_{j \in I_k} \pi^\ast_j}.
\end{equation}
This construction ensures that the probabilities of all rules with the same left-hand side sum to one, as required for a stochastic grammar.

\begin{figure*}[htbp]
\centering
\fbox{
\begin{minipage}[t]{0.43\textwidth}
    \section*{Grammar Rules}
    \vspace{-1mm}
    \[
    \begin{array}{rl}  
    I. & S \xrightarrow{p_{1}} d S \\
    II. & S \xrightarrow{p_{2}} D \\
    III. & D \xrightarrow{p_{3}} d \\
    IV. & S \xrightarrow{p_4} d SB \\
    V. & S \xrightarrow{p_5} d B \\
    VI. & B \xrightarrow{p_6} b \\
    VII. & S \xrightarrow{p_{7}} d SBC \\
    VIII. & S \xrightarrow{p_{8}} d BC \\
    IX. & dB \xrightarrow{p_{9}} d b \\
    X. & bC \xrightarrow{p_{10}} bc \\
    XI. & CB \xrightarrow{p_{11}} BC \\
    XII. & cC \xrightarrow{p_{12}} cc \\
    \end{array}
    \]

\vspace{-4mm}
\section*{Examples of Trajectories for Different Grammar}
\vspace{-4mm}

\begin{center}
\begin{tikzpicture}[>=Stealth, scale=1]

\begin{scope}[shift={(-1,0)}]
\node (A1) at (0,0) {};
\node (AB1) at (1,0) {};
\node (B1) at (2,0) {};
\node (BC1) at (1.5,0.9) {};
\node (C1) at (1,1.8) {};
\node (AC1) at (0.5,0.9) {};
\draw[->] (A1) -- (AB1);
\draw[->] (AB1) -- (B1);
\draw[->] (B1) -- (BC1);
\draw[->] (BC1) -- (C1);
\draw[->] (C1) -- (AC1);
\draw[->] (AC1) -- (A1);
\node at (1,-0.8) {$d^n b^n c^n$ };
\node at (1,-1.3) {(SCSG)};
\end{scope}

\begin{scope}[shift={(-3.5,0)}]
\node (A1) at (0,0) {};
\node (AB1) at (1,0) {};
\node (B1) at (2,0) {};
\node (BC1) at (1.5,0.9) {};
\node (C1) at (1,1.8) {};
\node (AC1) at (0.5,0.9) {};
\node (C2) at (0.5,2.7) {};
\draw[->] (A1) -- (AB1);
\draw[->] (AB1) -- (B1);
\draw[->] (B1) -- (BC1);
\draw[->] (BC1) -- (C1);
\node at (1,-0.8) {$d^n b^n$};
\node at (1,-1.3) {(SCFG)};
\end{scope}

\begin{scope}[shift={(-6,0)}]
\node (A1) at (0,0) {};
\node (AB1) at (1,0) {};
\node (B1) at (2,0) {};
\draw[->] (A1) -- (AB1);
\draw[->] (AB1) -- (B1);
\node at (1,-0.8) {$d^n $};
\node at (1,-1.3) {(SRG)};
\end{scope}

    \end{tikzpicture}
    \end{center}
\end{minipage}
\hfill
\begin{minipage}[t]{0.53\textwidth}
    \section*{Example of Trajectory Derivations}
    \vspace{-1mm}
    \[
    \begin{array}{rll}
    & dd & (I-III)\\
    & ddd & (I-III)\\
    & ddddd & (I-III)\\
    & ddbb & (III-VI)\\
    & dddbbb & (III-VI)\\
    & dddddbbb & (I-VI)\\
    & ddbbcc  & (VII-XII)\\
    & dddbbbcc  & (IV-XII)\\
    & dddddbbbcc & (I-XII)\\
    \end{array}
    \]

    \section*{Choice of Characteristic Functions in Cooperative Game}
    Consider we have three players (targets), player 1 is assigned with production rules I-III; player 2 is assigned with IV-VI; player 3 is assigned with production rules VII-XII. 
    \begin{enumerate}
        \item If we only want to surveillance the environment in a line shape, then simply the participation of player 1 in the game is enough, which leads to a stochastic regular grammar (SRG).
        \item If we only want to surveillance the environment in a corner shape, then the participation of both player 1 and player 2 in the game should be enough, which leads to a stochastic context-free grammar (SCFG).
        \item If we want to surveillance the environment in a triangular shape, then the participation of three players should be enough, which leads to a stochastic context-sensitive grammar (SCSG).
    \end{enumerate}
\end{minipage}
}
\caption{Relationship between grammar rules, example trajectory derivations, and characteristic functions for group intent in the cooperative game framework.
\textbf{Left:} Grammar production rules—SRG, SCFG, and SCSG—shown with their generative limitations.
\textbf{Right:} Example derivations and how characteristic functions select target subsets to realize specific group intents.}
\label{fig:grammar}
\end{figure*}
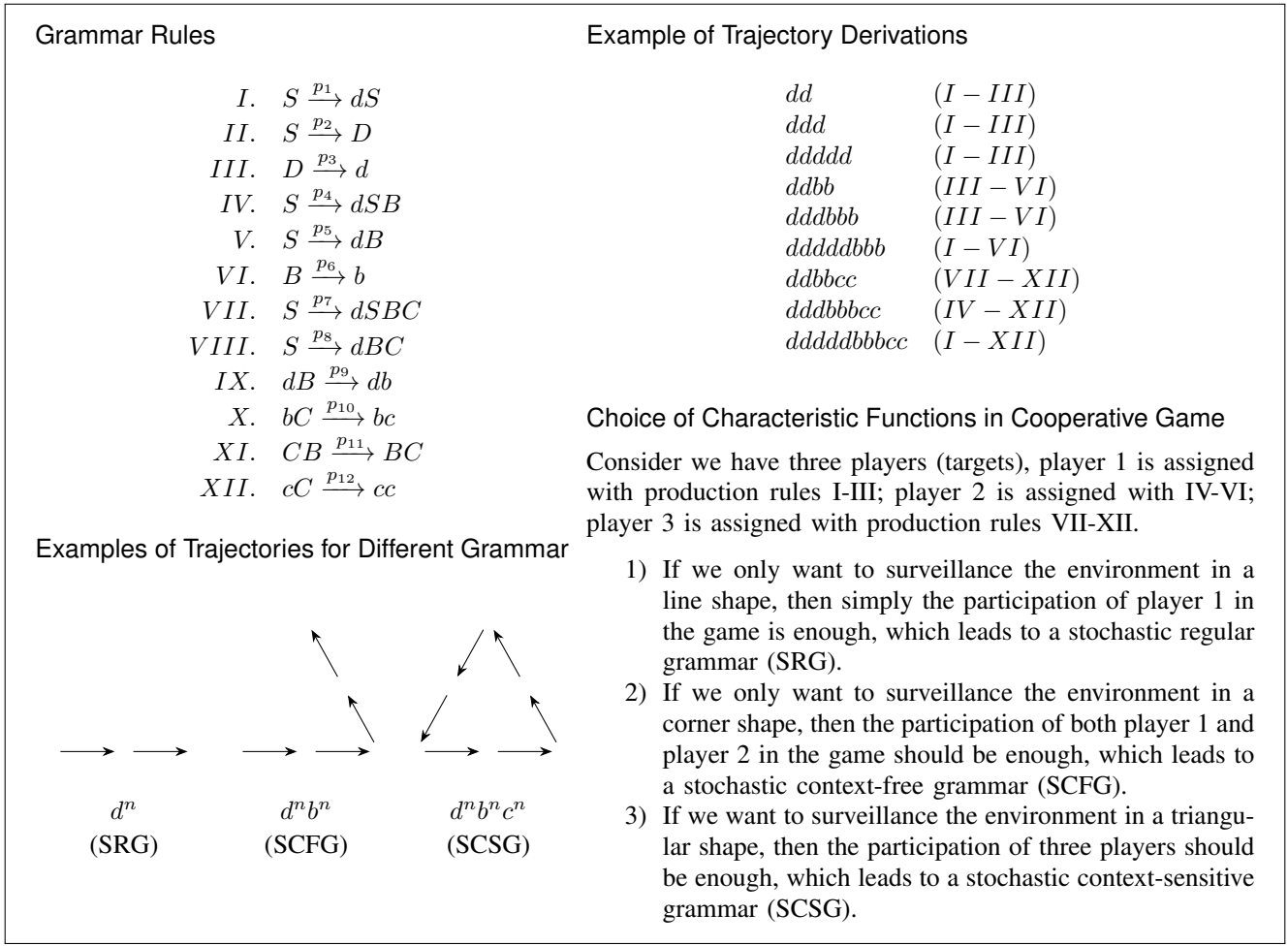

We connect the characteristic function to formal grammars, aiming for maximizers that correspond exactly to strings generated by \( G = (\mathcal{E}, \mathcal{A}, \Gamma, \mathcal{P}) \). A sufficient condition is  
\begin{equation}
\label{eq:v_limitation}
\utility(S \cup \{i\}) = \utility(S), \quad \forall S \subseteq N \setminus \{i\},
\end{equation}
ensuring that adding player \( i \) never alters coalition utility, so the cooperative game’s allocation cost matches the intent’s objective exactly. Theorem~\ref{theorem:zero_payoff} verifies that \eqref{eq:v_limitation} meets this requirement.

\begin{theorem}[Zero payoff in the core]
\label{theorem:zero_payoff}
Let $\mathcal{CG}(N, u)$ be a cooperative game with characteristic function $v: 2^N \to \mathbb{R}$, where $N$ is the set of players.  
If player $i \in N$ satisfies
\[
\utility(S \cup \{i\}) = \utility(S), \quad \forall S \subseteq N \setminus \{i\}.
\]
In any allocation $\pi$ in the core of $\utility$, it holds that $\pi_i = 0$.
\end{theorem}

\begin{proof}
By assumption, $\utility(S \cup \{i\}) = \utility(S)$ for all $S \subseteq N\setminus\{i\}$, so player $i$ has zero marginal contribution to any coalition.  
In particular, if $u(N) = u(N\setminus\{i\})$, then by Eq.~\eqref{eq:Core}, any core allocation must satisfy:
\begin{align*}
    &\sum_{j\in N\setminus\{i\}} \pi_j \geq u(N\setminus\{i\}) = u(N),\\
    &\sum_{j\in N} \pi_j = u(N).
\end{align*}
Combining the above expressions, and writing $\sum_{j\in N} \pi_j = \pi_i + \sum_{j\in N\setminus\{i\}} \pi_j$, we conclude that
\[
\pi_i = u(N) - \sum_{j\in N\setminus\{i\}} \pi_j \leq 0.
\]
On the other hand, by Eq.~\eqref{eq:Core}, $\pi_i \geq 0$.  
Therefore, $\pi_i = 0$.
\end{proof}
\noindent This result shows that a player whose participation never changes the game’s value receives zero payoff in the core, thus zero payoff in the nucleolus, which in our grammar-based framework means that the probability assigned to its associated production rule $P(r) = 0$ as well.

By analogy with Theorem~\ref{theorem:zero_payoff}, any production rule in a probabilistic grammar that contributes nothing beyond the forms allowed by the target grammar can have its probability set to zero without affecting generative capability. 
This allows us to enforce specific levels of the Chomsky hierarchy:
\begin{itemize}
    \item \textbf{Regular grammar:} This restricted form corresponds to the structure of a Hidden Markov Model (HMM). Set probabilities to 0 for all rules except those of the form $A \to aB$ or $A \to a$ (eliminating, for instance, $A \to BC$ rules with multiple nonterminals on the right-hand side).
    \item \textbf{Context-free grammar:} Set probabilities to 0 for all rules except those where the left-hand side is a single nonterminal, $A \to \gamma$ (eliminating context-dependent forms such as $\alpha A \beta \to \gamma$).
    \item \textbf{Context-sensitive grammar:} Keep all rules without setting any to zero, thus preserving maximum expressive power.
\end{itemize}
This framework directly ties the {cooperative game-based } utility constraints to the expressiveness control of the grammar. In Fig.~\ref{fig:grammar}, we illustrate a simple example demonstrating how this correspondence naturally emerges in practice. Specifically, to ensure the continuity~\footnote{Trajectories evolve continuously in space and time; targets do not appear or disappear instantaneously, nor do they teleport between locations.} and compatibility of trajectories assigned to different targets within a group, the trajectory sets are hierarchically constrained. That is, the set of trajectories that can be generated by one target is a subset of those generated by another target, and so on. 

For instance, consider a group of targets where one agent generates trajectories corresponding to the string $d^n$, another generates $d^n b^n$, and a third generates $d^n b^n c^n$. Each agent’s trajectory grammar extends that of the previous one, thereby maintaining continuity while allowing for increasing expressiveness or complexity in motion behavior. As shown in Fig.~\ref{fig:grammar}, all trajectories originate from the same initial point, ensuring spatial continuity and cooperative feasibility across the group.


\subsection{Fisher-information Based Characteristic Function}
\label{subsec:fisher}
As a concrete instantiation, we propose an explicit Fisher-information based characteristic function that not only satisfies the conditions of Theorem~\ref{theorem:zero_payoff} but also provides strong domain-specific insight in the context of  coordinated group behavior.

Consider the scenario where sensors are assigned to each target over a surveillance environment, each with potentially overlapping fields of view and heterogeneous sensing capabilities. 
Due to such spatial correlations and redundancy in target placements, certain target may offer no additional independent information about the environment beyond what is already captured by other target. 
For instance, a target positioned far from the region of interest, or entirely shadowed within the coverage of nearby sensors, contributes zero marginal information to the collective tracking performance. 
In such cases, Theorem~\ref{theorem:zero_payoff} implies that the payoff allocated to these sensors in the core is zero, and their associated production rules in the stochastic grammar receive zero probability. 
This motivates defining the coalition characteristic function in terms of mutual information between environmental measurements and the coalition of targets, enabling a principled analysis of conditions under which a target's marginal contribution is null. The characteristic function is defined as:
\begin{equation}
    \utility(S) = \mathrm{tr}(J_S)
\label{eq:fisher}
\end{equation}
where $S$ is the coalition of players, $J_S$ is the accumulative fisher information of the coalition, and $\mathrm{tr}$ denotes the trace of a matrix. The detailed formation of the characteristic function is as follows.
Let \( \zeta \in \mathbb{R}^d \) represent some surveillance parameter where \(j^{\mathrm{th}}\) target in the coalition  obtain $U$ measurements denoted using $m_j$ such that:
\[
m^u_j = H_j \zeta + \nu^u_j, 
\]
where, $u=1,\,2,\,...,\,U$.  \(H_j \in \mathbb{R}^{c \times d}\) is the observation matrix associated with target \(j\), and 
\(\nu_j\) is zero-mean Gaussian noise with covariance \(R_j \in \mathbb{R}^{c \times c}\).
Then the probability distribution of measurement $m^u_j$ given $\zeta$ is:
$$
p(m^u_j | \zeta) = \frac{\exp\left({-\frac{1}{2}(m^u_j-H_j \zeta)^\top R^{-1}_j(m^u_j-H_j \zeta)}\right)}{\sqrt{(2\pi)^c|R_j|}}
$$
Then the likelihood function of $\zeta$ is:
$$
L_j(\zeta;\;m_{j}) = \prod_{u=1}^{U} p(m^u_j | \zeta) 
$$
The log likelihood function is:
\begin{align*}
l_j(\zeta;\;m_{j}) &= -\frac{1}{2}\log\!\left({(2\pi)^m |R_j|}\right) \nonumber \\
&\quad - \frac{1}{2}\sum_{u=1}^{U}(m^u_j - H_j \zeta)^\top R_j^{-1}(m^u_j - H_j \zeta).
\end{align*}

Then the Fisher information of target $j$ given $\theta$ is defined as:
$$
\mathcal{I}_j(\zeta) = \mathbb{E}[\nabla_\zeta^2 l(\zeta) | \zeta] = U H_j^\top R_j^{-1} H_j.
$$
When a coalition \(S \subseteq N\) is formed, it gathers all measurements 
\(\{m_j : j \in S\}\).  
These measurements are conditionally independent given \(\zeta\), and the combined Fisher information matrix is
\[
J_S =  U \sum_{j \in S} H_j^\top R_j^{-1} H_j.
\]

We define the coalition's utility as the trace of the Fisher information matrix:
\begin{equation}
\label{eq:FIMtrace}
    \utility(S) = U\mathrm{tr}\left( \sum_{j \in S} H_j^\top R_j^{-1} H_j\right).
\end{equation}

Now consider adding a player \(i\) to coalition \(S\), forming \(S' = S \cup \{i\}\).  
The updated Fisher information matrix becomes:
\[
J_{S'} = J_S + U H_i^\top R_i^{-1} H_i.
\]
The marginal contribution of player \(i\) is therefore:
\[
\utility(S') - \utility(S) = \mathrm{tr}(J_S') - \mathrm{tr}(J_S)
\]
If the additional term \(H_i^\top R_i^{-1} H_i = 0\),  the marginal contribution is zero.  
This occurs, for example, when the sensor corresponding to target \(j\) is missing observations ($\mathrm{tr}(R_j)$ is very large) or it is placed where it cannot observe the region of interest.  In that case we have:
\[
\utility(S \cup \{i\}) - \utility(S) = 0, \quad \forall S \subseteq N \setminus \{i\}.
\]

By Theorem~\ref{theorem:zero_payoff}, such players receive zero payoff in the core, and their corresponding production rules can be assigned zero probability in the stochastic grammar without affecting performance. 
Furthermore, 
if the characteristic function $\utility$ is supermodular, the Shapley value of the cooperative game lies within the core~\cite{shapley1953value}, where we show in Theorem~\ref{theorem:trace_supermodular}. 
Thus, the core, the nucleolus, and the Shapley value all provide valid principles for determining the allocation among coalitions $S$.

\begin{theorem}[Trace-of-sum is supermodular]
\label{theorem:trace_supermodular}
The characteristic function $\utility$ defined in \eqref{eq:FIMtrace} is modular. By modular we mean that for all $S\subseteq T\subseteq N$ and $j\in N\setminus T$,
\[
\utility(S\cup\{j\})-\utility(S)\;=\;\utility(T\cup\{j\})-\utility(T)\;=\;\operatorname{tr}(M_j).
\]
\end{theorem}

\begin{proof}
By additivity of the matrix sum and linearity of the trace,
\[
\utility(S\cup\{j\})-\utility(S)
= \operatorname{tr}\!\Big(\sum_{i\in S\cup\{j\}}M_i\Big) - \operatorname{tr}\!\Big(\sum_{i\in S}M_i\Big)
= \operatorname{tr}(M_j),
\]
which does not depend on $S$. The same calculation with $T$ in place of $S$ gives
\[
\utility(T\cup\{j\})-\utility(T)=\operatorname{tr}(M_j).
\]
Hence the marginal gain is constant across contexts, so the supermodularity inequality
\[
\utility(S\cup\{j\})-\utility(S)\;\le\;\utility(T\cup\{j\})-\utility(T)
\]
holds with equality (and similarly the submodularity inequality). Therefore $\utility$ is modular, and in particular supermodular.
\end{proof}

In summary, Sec.~\ref{sec:game} establishes a bridge between grammar–based trajectory modeling and cooperative game theory, enabling intent to be encoded as characteristic function-driven rule allocations. This integration provides a mathematically grounded mechanism to capture both the syntactic structure and strategic coordination underlying complex group behaviors.

\section{Parse-tree Based Inference of Group Intent using GTNN}
\label{sec:classifier}
\begin{figure*}[t!]
    \centering
    \includegraphics[width=1.05\linewidth]{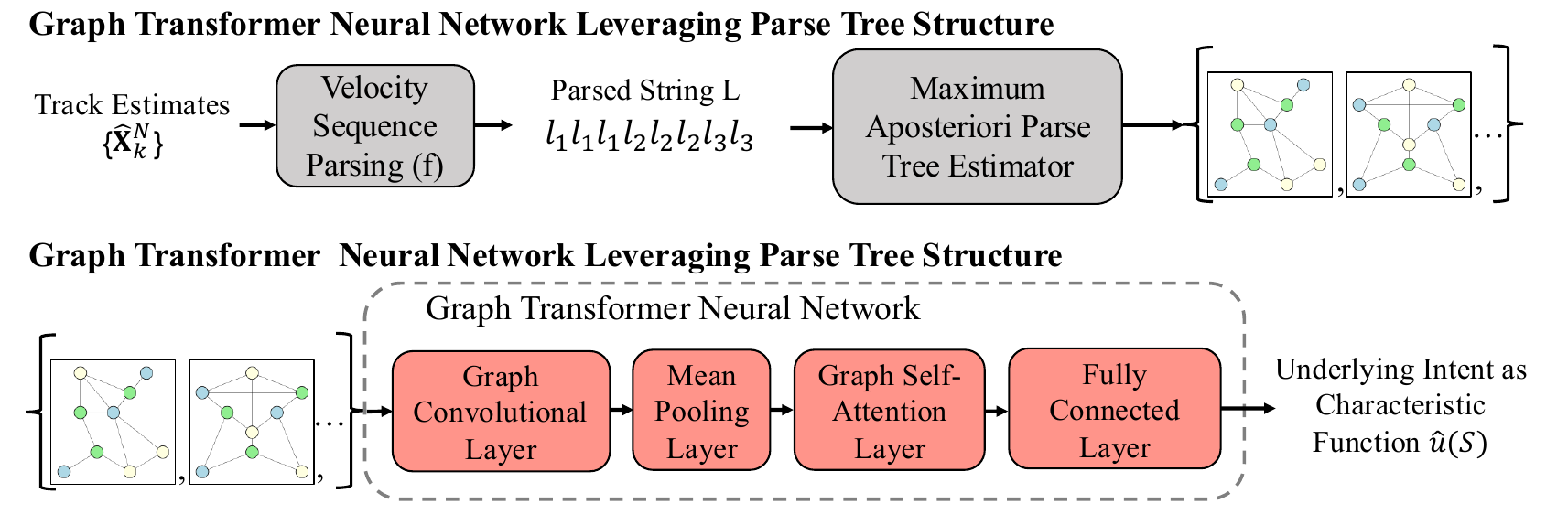}
    \caption{Overview of the proposed group intent analysis framework. The estimated states $\{\mathbf{\hat{X}}_k\}$ from the tracker are encoded using a point data encoder $f$, followed by maximum a posteriori (MAP) parse tree estimation using the CYK algorithm, which produces the parse tree $T$. The parse tree is fed into the \gls*{gtnn}, employing graph convolutional layer, mean pooling, graph self attention, and fully connected layers, to infer the characteristic function group intent.}
    \label{fig:design}
\end{figure*}

In this section, we present our second main result, namely an approach for inferring group intent, expressed as the characteristic function of a cooperative game in \eqref{eq:master}. We focus on stochastic context-sensitive grammars (SCSG), which subsume the expressive power of stochastic context-free (SCFG) and stochastic regular grammars (SRG); thus, a neural network trained on SCSG-based intent can also capture SCFG and SRG-driven intents. Section~\ref{sec:classifier}--\ref{subsec:rg_cfg} outlines inference methods for SRG and SCFG, Sec.~\ref{sec:classifier}--\ref{subsec:parsing} details the encoding of group trajectories into parse trees, and Sec.~\ref{sec:classifier}--\ref{subsec:neural_network}, Fig.~\ref{fig:design}, present the \gls*{gtnn} architecture for predicting the characteristic function from these parse trees.

\subsection{Stochastic Regular Grammars (SRG) and Stochastic Context Free Grammars (SCFGs)}
\label{subsec:rg_cfg}
SRGs (which are equivalent to Hidden Markov Models) extend classical regular grammars by associating probabilities with each production rule, enabling the modeling of both structural constraints and statistical tendencies in sequential data \cite{manning1999foundations,hopcroft2006automata}. Inference for SRGs is often formulated in terms of probabilistic finite-state automata, where algorithms such as the forward–backward procedure or the Baum–Welch algorithm are employed to compute sequence likelihoods and optimize rule probabilities \cite{baum1970maximization,juang1991hidden}. These methods are computationally efficient, making SRGs suitable for domains where the underlying structure is shallow and non-hierarchical.

Stochastic context-free grammars (SCFGs) generalize this approach to context-free grammars, allowing for recursive and nested structures \cite{booth1973applicability}. Exact inference in SCFGs is typically performed via probabilistic parsing algorithms, most notably the inside–outside algorithm \cite{baker1979trainable}, or probabilistic variants of the \gls*{cyk} algorithm \cite{Kasami1965AnER}. For large or complex grammars, approximate inference strategies such as Monte Carlo sampling over parse trees \cite{johnson2007adapting} or variational inference \cite{kurihara2006bayesian} are used to improve scalability.

However, the Bayesian inference for \gls*{scsg} is NP-hard. To mitigate the computational challenge, we represent the 
trajectories generated by the
syntactic rules of the SCSG 
as parse trees which allow us to exploit the capabilities of graph neural network for efficient inference of \gls*{scsg}s, which we will introduce in Sec.~\ref{sec:classifier}-\ref{subsec:parsing} and \ref{sec:classifier}-\ref{subsec:neural_network}.

\subsection{Trajectory Encoding and Grammar Tree Parsing}
\label{subsec:parsing} In this subsection, we present the trajectory as a parse tree representation, which then serves as the input embedding for the \gls*{gtnn}.

A straightforward approach to classifying the state sequence $\{\mathbf{\hat{X}}^N_k\}$ from \eqref{eq:track_est} would be to train a deep classifier that directly takes $\{\mathbf{\hat{X}}^N_k\}$ value sequences as input. However, this naive strategy overlooks the rich hierarchical structure encoded in the parse tree of a \gls*{scsg}, which can provide valuable contextual information for more accurate and interpretable classification.
In this section, we exploit the parse tree structure to classify target intent, as illustrated in Fig.\,\ref{fig:design}. Our ``grammar aware" approach involves converting trajectory data into a parse tree and inferring the group intent.

\noindent{\bf Velocity Sequence Parsing:} The estimated states $\{\mathbf{\hat{X}}^N_k\}$ form a sequence of groups of $N$ target states at timestep $k$, capturing the variation between timesteps. Each $\{\mathbf{\hat{X}}^N_k\}$ is then mapped to a terminal symbol in the \gls*{scsg} using the parsing approach described in Sec.~\ref{sec:grammar}-\ref{subsec:IIB},
\begin{equation}
\label{eq:terminal_symbol}
L = f({\mathbf{\hat{X}}^N_k}),
\end{equation}
where $f$ denotes the parsing function. Consequently, the entire trajectory is encoded as a string
$L = l_1l_2\ldots l_{k}$, where each $l_k$ is a terminal symbol from the alphabet set $\mathcal{A}$ of the \gls*{scsg} $G$ defined in Sec.~\ref{sec:grammar}-\ref{subsec:grammar_modeling}. This symbolic sequence provides a compact representation of the trajectory, enabling efficient parsing and clustering through the grammatical structure of the \gls*{scsg}.

\noindent{\bf Grammar Tree Parsing: }
Given a generated string $L$, we derive its structural parse tree $T$ by employing a chart-based algorithm in the style of CYK, adapted to Linear Context-Free Rewriting Systems (LCFRS)~\cite[Ch.~3]{seki1991multiple}. This extension allows us to handle discontinuous constituents while maintaining polynomial-time complexity ($O(n^6)$), unlike the intractability of full context-sensitive parsing where derivation tracking is PSPACE-complete. To ensure tractability, our parsing is restricted to this mildly context-sensitive class, with non-terminal spans bounded to pairs. Under these constraints, we approximate the derivation tree as  

\begin{equation}
\label{eq:parse_tree}
    T = \text{Parse}(L, G), 
\end{equation}
where $G$ is a mildly context-sensitive grammar represented as an LCFRS, yielding a framework that is both expressive enough for natural language phenomena and computationally feasible.

\subsection{Characteristic Function Inference via GTNN}
\label{subsec:neural_network}

\subsubsection{GTNN Architecture}

The \gls*{gtnn} architecture is illustrated in Fig.~\ref{fig:design}. 
The self-attention mechanism in the \gls*{gtnn} is explicitly obtained from the underlying track estimates as will be specified in~\eqref{equation:alpha} below. 
The hierarchical tree structure $T = (V,E)$, obtained from grammar-based parsing of target trajectories by \eqref{eq:parse_tree}, is used to infer a set of parameters from which the characteristic function $\utility(S)$ of a cooperative game can be computed. Specifically, we employ a graph neural network based on the Graph Transformer architecture~\cite{yun2019graph} to map $T$ into a coalition-aware representation. 

The Graph Convolutional Network (GCN) in the \gls*{gtnn} serves as the first stage for extracting node-level representations from the parse tree $T = (V, E)$. Each node $i \in V$ is initialized with a feature vector $h_i^0 \in \mathbb{R}^d$, which encodes structural attributes such as degree, in-degree, and clustering coefficient, describing its local role within the graph. The GCN is composed of multiple stacked convolutional layers, where each layer aggregates information from a node’s local neighborhood to capture both immediate and higher-order dependencies. The feature update at the $l$-th layer is defined as
\[
h_i^{l+1} = \mathrm{ReLU}\!\left(\sum_{j \in \mathcal{N}(i) \cup \{i\}} \frac{1}{c_{ij}}\, W^{l} h_j^{l}\right),
\]
where $h_j^{l}$ denotes the feature of node $j$ at layer $l$, $W^{l}$ is a trainable weight matrix, $c_{ij}$ is a normalization coefficient based on node degrees, and $\mathcal{N}(i)$ represents the neighborhood of node $i$. Through successive layers, each node embedding $h_i^{L}$ integrates information from increasingly larger neighborhoods, resulting in representations that capture rich structural and contextual information to be processed by the subsequent transformer module.

The node embeddings produced by the GCN are aggregated into a single, fixed-size graph-level representation through a mean pooling operation. This step ensures permutation invariance and enables graphs of varying sizes to be represented consistently. For the graph with node set $\mathcal{V}_b$, the pooled representation is computed as  
\[
h_b = \frac{1}{|\mathcal{V}_b|} \sum_{i \in \mathcal{V}_b} h_i^{L},
\]
where $h_i^{L}$ denotes the final embedding of node $i$ after the last GCN layer. The resulting vector $h_b$ provides a compact summary of the graph’s structural and semantic characteristics and serves as the input token for the subsequent transformer encoder.

The transformer module in the \gls*{gtnn} takes as input the pooled graph embeddings $h_b$ obtained from the previous stage and models their global dependencies through a multi-head self-attention mechanism. This component enables the network to capture contextual interactions between graphs or higher-level structures that cannot be learned through local message passing alone. The self-attention layer computes an attention-weighted representation for each input embedding as  
\begin{equation}
\label{equation:alpha}
    \tilde{h}_b = \mathrm{softmax}\!\left(\frac{Q_b K_b^{\top}}{\sqrt{d_k}}\right)V_b,
\end{equation}
where $Q_b = h_b W_Q$, $K_b = h_b W_K$, and $V_b = h_b W_V$ are the query, key, and value projections of $h_b$, respectively, and $d_k$ is the key dimension. The resulting output $\tilde{h}_b$ encodes global contextual information across all graph embeddings, forming a refined representation that is subsequently processed by the final fully connected layer.

{\em Remark: From Tracker to Self-attention.} Note  $h_i^0$ depends on the parse-tree \eqref{eq:parse_tree}, which in turn depends on the track estimates \eqref{eq:bayes}. Therefore, \eqref{equation:alpha} relates the track estimates to the self attention mechanism of the GTNN.

Finally, the \textbf{Fully Connected (Dense) layer} transforms the transformer output $\tilde{h}_b$ into the final prediction vector $\theta$. This mapping is expressed as  
\[
\theta = g(W_{\text{out}} \tilde{h}_b + b_{\text{out}}),
\]
where $W_{\text{out}}$ and $b_{\text{out}}$ are learnable parameters, and $g(\cdot)$ denotes a nonlinear activation function such as ReLU. In this formulation, $\theta$ represents the graph-level output of the \gls*{gtnn}, capturing both the local structural features extracted by the GCN and the global contextual dependencies modeled by the transformer. Each component of $\theta$ corresponds to a learned embedding associated with a specific graph or agent, serving as the basis for evaluating the cooperative-game characteristic function $\utility(S)$ across different coalitions.



\subsubsection{Inference of Group Intent as Characteristic function of Cooperative Game}
In this work, the neural network transforms the (what are these representations) global graph representation into the vector $\theta$ that parametrizes the characteristic function $\utility(S)$. 
This stage can learn complex nonlinear mappings, enabling the model to reorganize and combine features extracted from the graph in a way that directly supports coalition value prediction.
These updated node embeddings are aggregated using a permutation-invariant pooling operation to produce a vector $\theta \in \mathbb{R}^{N}$, where each item $a_i$ represents a learned embedding for player $i$ in the game.
For any coalition $S \in \{0,1\}^N$,  we define $\hat{v}(S)$ using a generic form of characteristic function:
\begin{equation}
\hat{u}(S) = \theta^\top \!\left(S + \sigma(\,\mathit{\Lambda} S )\right),
\end{equation} 
where $\mathit{\Lambda} \in \mathbb{R}^{N \times N}$ is a fixed parameter, and $\sigma(\cdot)$ denotes the sigmoid activation function. This construction preserves the additive contributions of individual players (through $S$) while introducing nonlinear corrections (through $\sigma(\mathit{\Lambda} S + \mathit{\delta})$) that capture richer coalition-level interactions. Consequently, the characteristic function can model complex synergies and dependencies among players that extend beyond purely linear effects. 


The model can evaluate all $2^N$ possible coalitions in each forward pass. The training loss is then defined as the mean squared error over all coalitions:
\begin{equation}
\label{eq:loss}
    L_{\mathrm{MSE}}= \frac{1}{2^N} \sum_{ S \in \{0,1\}^N} \left(\hat{u}_{}(S) - u_{\mathrm{}}(S)\right)^2,
\end{equation}
where $\utility(S)$ is the ground-truth characteristic function value for coalition $S$. Gradients from this loss are back-propagated through both the coalition-value mapping and the Graph Transformer, enabling end-to-end learning of player embeddings that fully encode the coalition structure.

By utilizing symbolic encoding and grammar-based parsing to extract structured inputs, and by training the Graph Transformer so that its output directly parametrizes $\utility(S)$ for all coalitions, our approach yields a model that is both expressive and faithful to cooperative game semantics, outperforming flat-sequence baselines in capturing the group intent.

\section{Numerical Results}
\label{sec:experiments} 
In this section, we discuss our numerical experiments and compare our proposed methods with  baseline methods. 
We compare our graph neural network with our baselines and show that the graph neural network that exploits the parse-tree structure of achieve better performance in predict the characteristic function as the group intent. Furthermore, we show that our methods also achieve better prediction performance in the domain of SCFG and SRG (HMMs), not simply in SCSG.

\textbf{Synthetic Dataset Generation\footnote{For the  reproducibility of the results, the codes have been uploaded in \href{https://github.com/yimingz1218/SCSG}{GitHub}.}.} 
The data-generation process is governed by production rules defined through characteristic functions. Specifically, we define ten distinct characteristic functions, each associated with its own set of production rules. 
Each data sample is represented as a tuple $(L, \utility(\cdot))$, where $L$ is a string derived from the velocity sequence, as illustrated in Sec.~\ref{sec:grammar}-\ref{subsec:IIB}, and $\utility(\cdot)$ is the characteristic function encoding the production rules that generate the corresponding trajectory. For every sample, a specific characteristic function is assigned, which probabilistically governs how the trajectory is produced according to its associated production rules.
The training set consists of 50000 samples (5000 for each distinct $\utility(\cdot)$, and the test set contains 5000 samples (500 for each distinct $\utility(\cdot)$. Each sample is an SCSG sequence  with the ground-truth characteristic function $\utility(S)$ determining the underlying probability distribution.  

\DIFaddbegin \textbf{\DIFadd{Model Parameters.}}
\DIFadd{For model training, we use the Adam optimizer with a learning rate of 0.001. 
The training objective is defined by the loss function~\eqref{eq:loss}. 
The model is trained for 5 epochs.
The model follows a hybrid GCN-Transformer architecture, where node features are first encoded using a Graph Convolutional Network (GCN) and then processed by a Transformer encoder. 
The model uses an input feature dimension of 5, hidden size of 128, two Transformer encoder layers with 4 attention heads, and an output dimension of 5.
}

\DIFaddend To assess the performance of our graph neural network for group intent inference, we define the \emph{success rate} $\kappa$ as the proportion of test cases in which the predicted characteristic function attains a loss~\eqref{eq:loss} that is sufficiently small relative to the ground truth. Formally, 
\begin{equation}
\label{eq:eva}
    \kappa = \frac{1}{\mathrm{M}} \sum_{j=1}^{\mathrm{M}} \mathbb{I}_{\{\!L_{\mathrm{MSE}, j} \leq \eta \}},
\end{equation}
where $\mathbb{I}_{\{\cdot\}}$ denotes the indicator function, which equals $1$ if the condition inside the brackets holds and $0$ otherwise. The term $L_{\mathrm{MSE}, j}$ represents the mean squared error for test case $j$, and $\mathrm{M}$ is the total number of test cases.
Thus, only predictions with losses below the threshold~$\eta$ contribute positively to the success rate. 

The \emph{sensor error probability} $q$ characterizes the level of noise from both \DIFdelbegin \DIFdel{both }\DIFdelend the target dynamics and the observation process in~\eqref{eq:statespace} and~\eqref{eq:observe}. 
Rather than perturbing the production rules directly, which would compromise the structural consistency of the grammar, $q$ is incorporated through a noisy terminal mechanism. In this formulation, $q$ governs the probability that a terminal symbol is replaced or corrupted during generation, thereby providing a principled means of modeling stochastic observation noise while preserving the integrity of the underlying grammar. Specifically, $q$ controls the likelihood that the grammar introduces random perturbations at the terminal level. In a stochastic grammar, $\mathcal{P}(r)$ denotes the probability of a production rule $r$, while $\mathcal{Q}(r)$ represents the set of all rules sharing the same left-hand side as $r$, as defined in equation~\eqref{eq:norm_prob}.

To incorporate perturbation, we augment the original set $\mathcal{Q}(r)$ with a special noise rule $r'$, yielding the updated distribution $\mathcal{P}_{Q'}(r)$ over the new set $\mathcal{Q'}(r)$. Each production step is then modified as follows: with probability $1-q$, the grammar samples from the original distribution and applies a standard rule; with probability $q$, the expansion is replaced by the noise symbol $\epsilon$.
Formally, for any production rule $i$, this is expressed using the indicator function $\mathbb{I}{{\cdot}}$, which equals $1$ if the condition inside holds and $0$ otherwise:
\begin{align}
\label{eq:Perturbation}
\mathcal{P}_{\mathcal{Q}'}(i) = (1-q)\,\mathbb{I}_{\{i \in \mathcal{Q}(r)\}}\mathcal{P}(i) + q\,\mathbb{I}_{\{i \notin \mathcal{Q}(r)\}}.
\end{align}
This construction rescales the probabilities of all original rules by a factor of $(1-q)$ and assigns the remaining probability mass $q$ to the special noise symbol $\epsilon$. As a result, the total probability remains normalized, and the system continues to define a valid stochastic grammar\footnote{The noise parameter 
$q$ depends on both the process noise in \eqref{eq:statespace} and the observation noise in \eqref{eq:observe}. However, it also alters the grammar rules, making it a more general source of uncertainty that extends beyond the process and observation noise. }.






\textbf{Baseline Models for Comparison.}
We benchmark our graph neural network approach for group intent inference against \DIFdelbegin \DIFdel{other }\DIFdelend \DIFaddbegin \DIFadd{several representative }\DIFaddend baselines. These include: DeepeST~\cite{icaart21}, an LSTM-based model designed for \DIFdelbegin \DIFdel{spatio temporal }\DIFdelend \DIFaddbegin \DIFadd{spatio-temporal }\DIFaddend sequence modeling; TraClets~\cite{KONTOPOULOS2023101306}, a vision-based model that converts trajectories into 2D representations for CNN inference; and XGBoost~\cite{chen2016xgboost}, a gradient-boosted decision tree classifier that processes structured state vector inputs. \DIFdelbegin \DIFdel{These }\DIFdelend \DIFaddbegin \DIFadd{In addition, we consider trajectory representation learning approaches such as the multilevel attention framework for driver identification ATTraj2vec ~\mbox{
\cite{li2025trajectory} }\hskip0pt
and a deep semi-supervised learning model for transportation mode detection DSM~\mbox{
\cite{sadeghian2024deep}}\hskip0pt
. These methods emphasize learned trajectory embeddings and hierarchical feature extraction for mobility understanding tasks. Collectively, the }\DIFaddend baselines span three major families of group intent classifiers—sequential, image-based, \DIFdelbegin \DIFdel{and tabular}\DIFdelend \DIFaddbegin \DIFadd{tabular, and representation-learning-based}\DIFaddend —providing comprehensive coverage of common \DIFdelbegin \DIFdel{strategies}\DIFdelend \DIFaddbegin \DIFadd{modeling strategies for trajectory-driven inference}\DIFaddend .
As shown in Fig.~\ref{fig:sota}, the grammar-aware graph neural network outperforms all baselines across varying sensor error probability values. While all models experience reduced accuracy \DIFdelbegin \DIFdel{under increasing noise }\DIFdelend \DIFaddbegin \DIFadd{as noise increases}\DIFaddend , our method \DIFdelbegin \DIFdel{maintains higher accuracy compared to other methods in }\DIFdelend \DIFaddbegin \DIFadd{consistently maintains higher performance across }\DIFaddend different sensor error \DIFdelbegin \DIFdel{probability values}\DIFdelend \DIFaddbegin \DIFadd{settings}\DIFaddend . This advantage stems from \DIFdelbegin \DIFdel{our }\DIFdelend \DIFaddbegin \DIFadd{the }\DIFaddend model’s ability to \DIFaddbegin \DIFadd{explicitly }\DIFaddend encode structural dependencies in \DIFdelbegin \DIFdel{the }\DIFdelend grammar-based parse trees and \DIFdelbegin \DIFdel{represent them through relational reasoning within the graph neural network.
}\DIFdelend \DIFaddbegin \DIFadd{capture relational interactions through graph-based reasoning, making it more robust to sensor perturbations and structural uncertainty.
In addition to the accuracy comparison, we also report the mean squared error (MSE) results in Fig.~\ref{fig:sota_mse}. The MSE comparison provides a complementary perspective on prediction quality by measuring the deviation between the predicted characteristic values and the ground truth. Consistent with the accuracy results, our grammar-aware graph model achieves lower MSE than the competing baselines across different sensor noise levels, further demonstrating its robustness and reliability.
}

\DIFaddend \begin{figure}[h!]
    \centering
    \includegraphics[width=\linewidth]{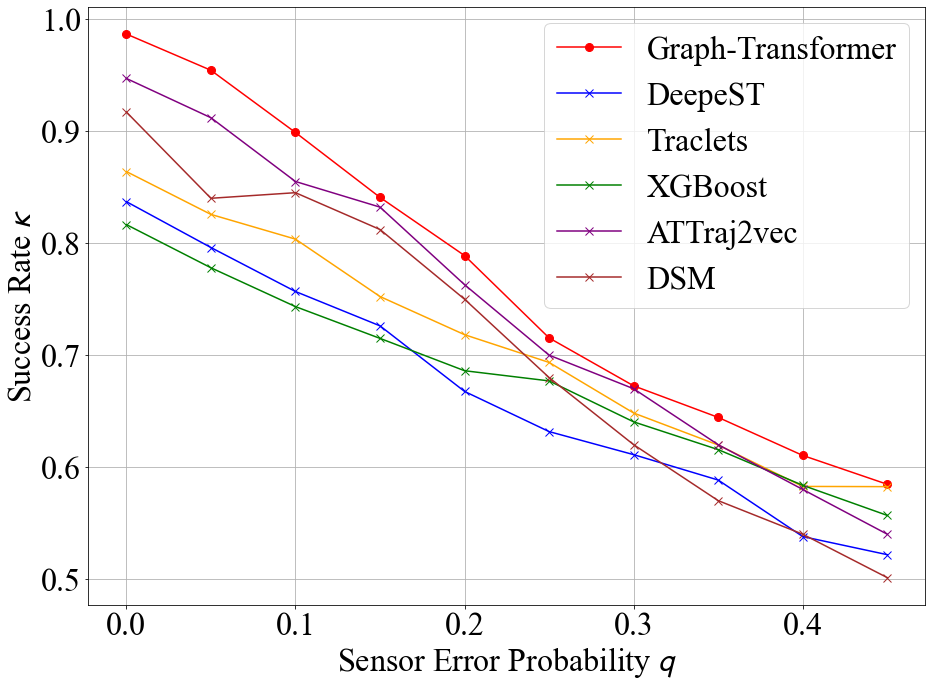}
    \caption{Comparison of model \DIFdelbeginFL \DIFdelFL{accuracy}\DIFdelendFL \DIFaddbeginFL \DIFaddFL{success rate}\DIFaddendFL , \eqref{eq:eva}, across sensor error probability values, \eqref{eq:Perturbation}, for different models: our proposed SCSG-aware \DIFdelbeginFL \DIFdelFL{GNN}\DIFdelendFL \DIFaddbeginFL \DIFaddFL{GTNN}\DIFaddendFL , DeepST, \DIFdelbeginFL \DIFdelFL{and }\DIFdelendFL XGBoost\DIFaddbeginFL \DIFaddFL{, Traclets, ATTraj2vec, and DSM}\DIFaddendFL . We show the superior performance of \DIFdelbeginFL \DIFdelFL{GNN }\DIFdelendFL \DIFaddbeginFL \DIFaddFL{GTNN }\DIFaddendFL across varying sensor error probability values in group intent inference.}
    \label{fig:sota}
\end{figure} 

 \begin{figure}[h!]
    \centering
    \includegraphics[width=\linewidth]{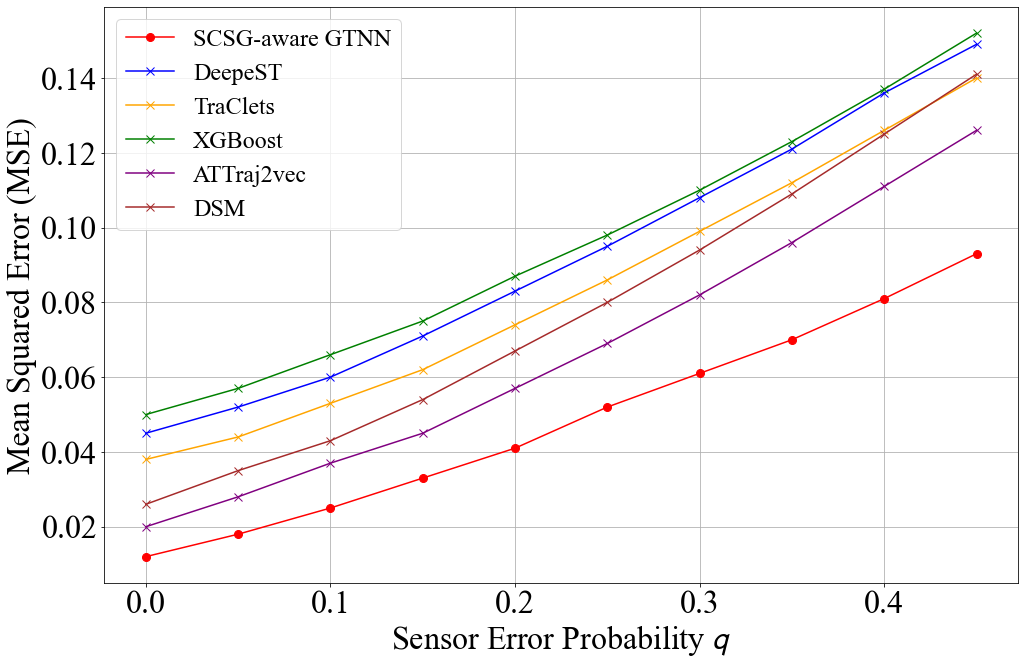}
    \caption{\DIFaddFL{Comparison of the mean squared error (MSE) between the estimated characteristic function and the ground truth, as defined in \eqref{eq:loss}, across different sensor error probability values in \eqref{eq:Perturbation}, for our proposed SCSG-aware GTNN and the baseline models DeepST, XGBoost, TraClets, ATTraj2vec, and DSM. The results show that GTNN consistently achieves lower MSE under varying sensor noise levels, demonstrating its superior robustness and effectiveness for group intent inference.}}
    \label{fig:sota_mse}
\end{figure}

\textbf{Comparison with Text-Based Models.}
 Fig.~\ref{fig:aba} compares our grammar-tree-based graph models (Graph-CNN, Graph-LSTM, Graph-Transformer) with their corresponding text-based versions (Text-CNN, Text-LSTM, Text-Transformer). Results clearly indicate the superiority of the graph-based models, which utilize the hierarchical parse-tree structure to preserve the syntactic and semantic properties of group intent, under increasing sensor error probability values. The graph representations allows the model to encode structural dependencies in the grammar-based parse trees and represent them through relational reasoning within the graph neural network. This confirms that directly predicting from strings is suboptimal compared to graph-structured inputs derived from our parsing step.
\DIFaddbegin \DIFadd{In particular, the graph-based models are better able to maintain stable performance as the sensor error probability increases, indicating stronger robustness to noise and perturbations in the observed sequences. This advantage arises because the parse tree explicitly organizes the compositional relationships among symbols, allowing the model to exploit both local and global structural information rather than relying only on sequential token patterns. 
These results further demonstrate that the proposed parsing step is not merely a preprocessing procedure, but a key component that exposes the latent grammatical structure needed for accurate characteristic-function prediction.
}\DIFaddend 

\begin{figure}[h!]
    \centering
    \includegraphics[width=\linewidth]{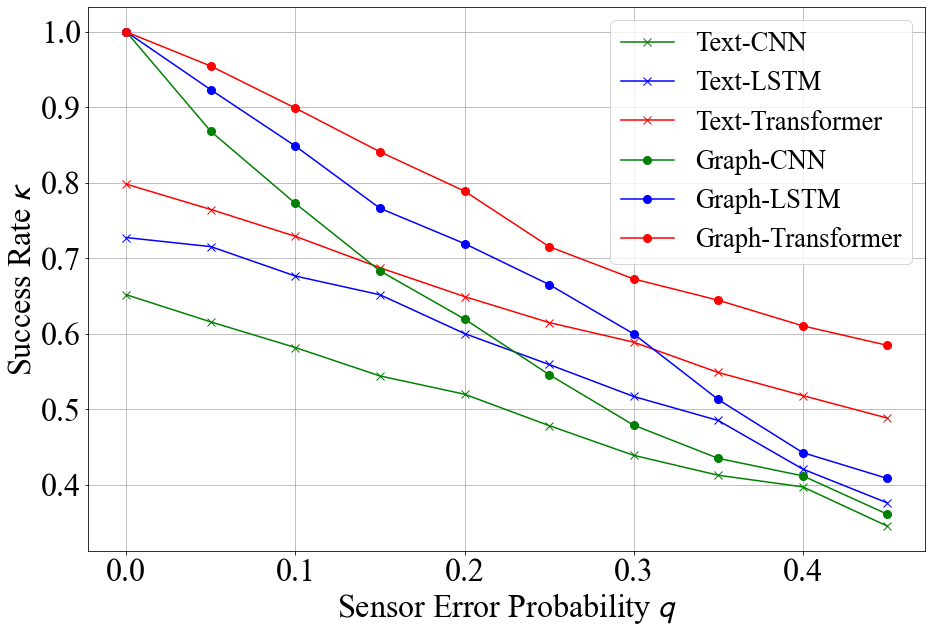}
    \caption{Comparison of model \DIFdelbeginFL \DIFdelFL{accuracy}\DIFdelendFL \DIFaddbeginFL \DIFaddFL{success rate}\DIFaddendFL , \eqref{eq:eva}, across sensor error probability values, \eqref{eq:Perturbation}, for text-based and graph-based architectures. We show that the proposed graph-based models consistently outperform text-based models in group intent inference.}
    \label{fig:aba}
\end{figure}

\textbf{Generalization to SCFG and SRG.}
 In Fig.~\ref{fig:more}, we further test the generalizability of our method in group intent inference by applying it to datasets governed by SCFG (Stochastic Context-Free Grammar) and SRG (Stochastic Regular Grammar), in addition to the default SCSG (Stochastic Context-Sensitive Grammar). 
Across all scenarios, the SCSG-based classifier still performs best, due to its ability to express richer dependencies in group behavior. However, our method also achieves competitive accuracy in SCFG and SRG settings, outperforming their respective grammar-specific baselines. This illustrates that our approach is not limited to a specific grammar formalism; instead, it generalizes well to different levels of linguistic complexity, reaffirming the versatility and scalability of our grammar-aware, graph-based classification framework.

\begin{figure}[h!]
    \centering
    \begin{subfigure}[b]{0.5\textwidth}
        \centering
        \includegraphics[width=\linewidth]{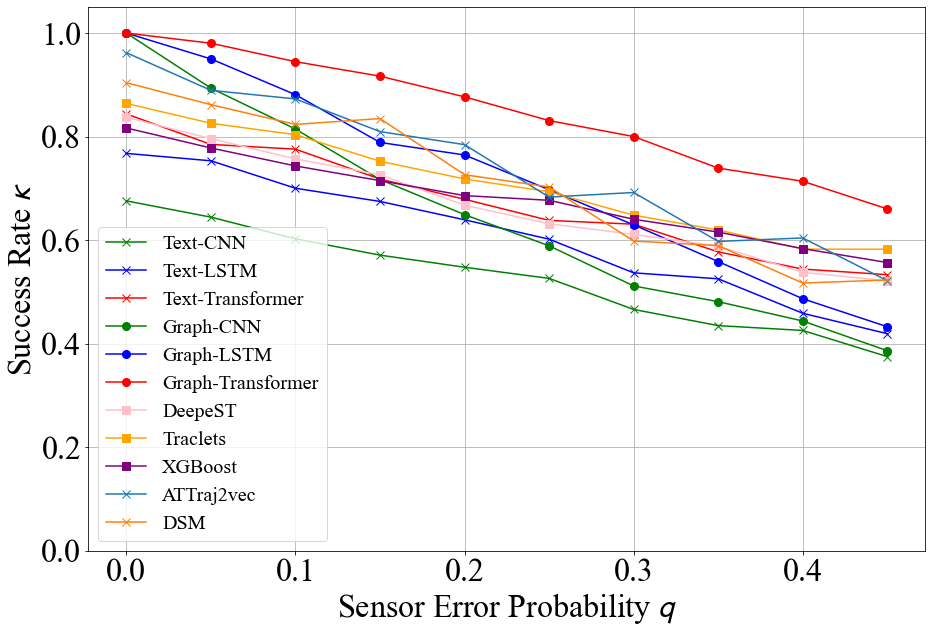}
        \caption{Performance compared using SCFG based production rules.}
        \label{fig:sub1}
    \end{subfigure}
    \vspace{2mm}
    \begin{subfigure}[b]{0.5\textwidth}
        \centering
        \includegraphics[width=\linewidth]{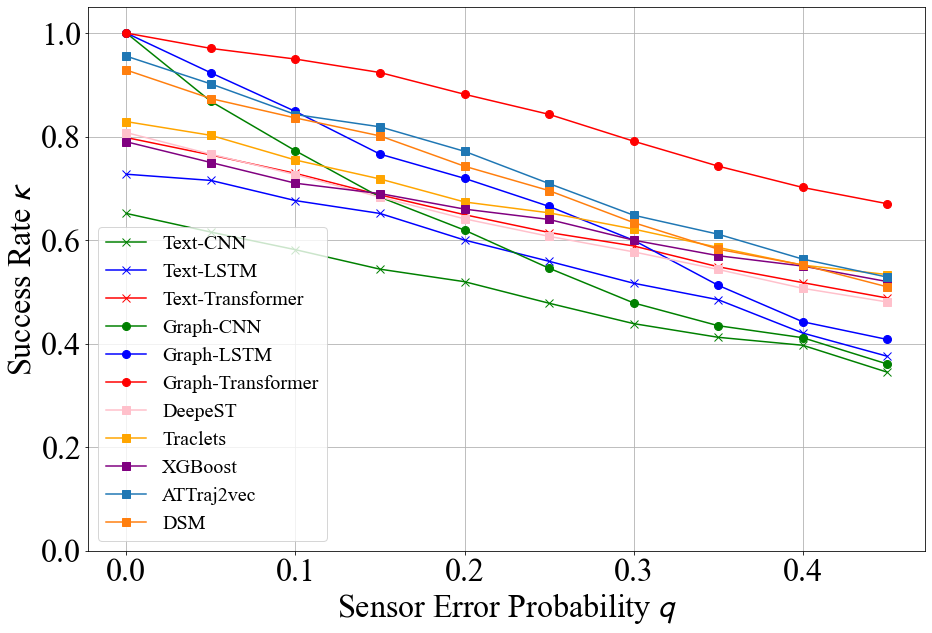}
        \caption{Performance compared using SRG based production rules.}
        \label{fig:sub2}
    \end{subfigure}
    \caption{Comparison of model accuracy, \eqref{eq:eva}, across different sensor error probability values, \eqref{eq:Perturbation}, for SRG and SCFG. The graph-based neural network consistently exhibits higher accuracy than other baselines, indicating greater certainty in group intent inference.}
    \label{fig:more}
\end{figure}

 {\textbf{Discussion on the Number of Targets.}
In our current experimental setting, the number of targets $N$ is fixed across samples so that the characteristic function of the cooperative game is defined over a consistent coalition space. To further examine the effect of group size, we additionally conducted experiments in which $N$ was varied while the underlying group intent was kept unchanged. The results, shown in Fig.~\ref{fig:varying_N}, indicate that the prediction success rate decreases as $N$ increases. This is expected, since a larger number of targets induces a larger coalition space and more complex grammar structures, which in turn lead to more challenging parse-tree-based intent inference.

\begin{figure}[h!]
    \centering
    \includegraphics[width=0.9\columnwidth]{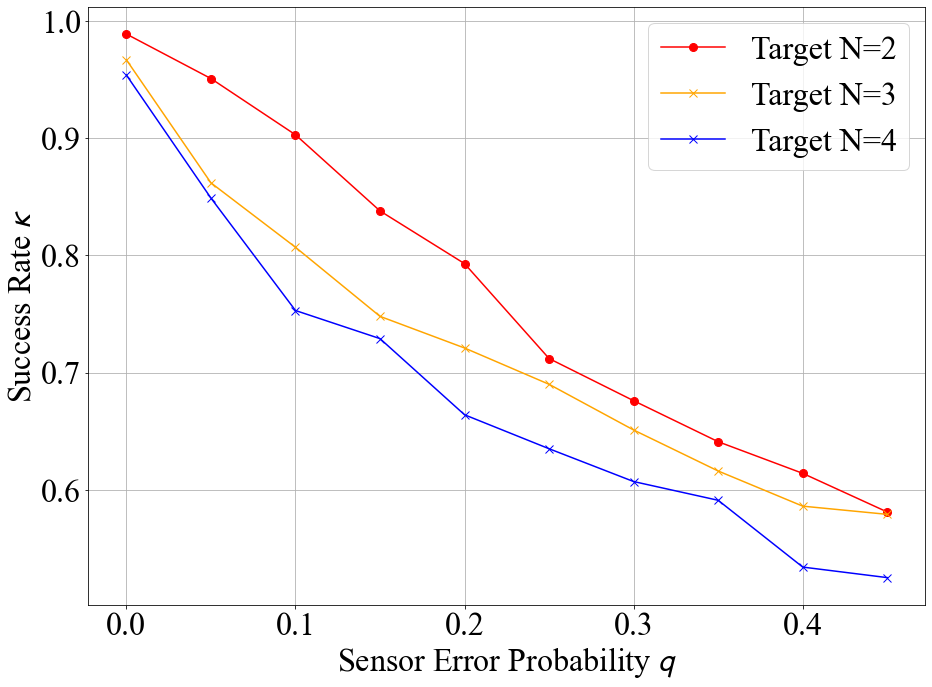}
    \caption{Comparison of model success rate versus the number of targets $N$ when the underlying group intent is kept fixed. The results show that the success rate decreases as $N$ increases, indicating that larger group sizes lead to more complex grammar and parse-tree structures, which make intent inference more challenging.}
    \label{fig:varying_N}
\end{figure}
}

\section{Conclusion}
\label{sec:conclusion}

This paper presented three main results: 

First, we formulated group intent as the outcome of a cooperative game. By specifying the characteristic function of the game in terms of the trace of the Fisher information matrix, the game inherits a structure that is amenable to efficient computation. Moreover, the core, Shapley value and nucleolus of the game provide  useful interpretations of group intent. 

Second, the outcome of the cooperative game was  used to specify the probabilities of the trajectory evolution of the targets using natural language models. These meta-level models fit seamlessly on top of a classical kinematic target tracking model, and serve as generative models of complex spatio-temporal trajectories of the targets, which cannot be captured by classical state space  models.

Third, we proposed a novel \gls*{gtnn} architecture to recover the characteristic function of the game, and therefore the intent. The proposed \gls*{gtnn} architecture exploits the grammatical structure of the trajectories. This ``grammar-aware" transformer demonstrated strong predictive accuracy, especially under noisy conditions and across different grammar classes (SCSG, SCFG, and SRG), outperforming baselines.

To summarize, we construct a model and a signal-processing intent inference methodology which spans from Bayesian Tracking to self-attention layer in transformer neural networks for group intent inference.

For future work, we will make extensions to online tracking, adversarial and deceptive behavior modeling, and applications in heterogeneous multi-agent domains such as autonomous driving and robotic swarms.

\bibliographystyle{ieeetr}
\bibliography{strings,refs}

\end{document}